\documentclass[lettersize,journal]{IEEEtran}
\usepackage{amsmath,amsfonts}
\usepackage{algorithmic}
\usepackage{algorithm}
\usepackage{array}
\usepackage{textcomp}
\usepackage{url}
\usepackage{verbatim}
\usepackage{graphicx,caption,subcaption,balance,stfloats}
\usepackage{cite}
\usepackage[english]{babel}

\usepackage{ifpdf}

\usepackage{cite} 
\usepackage{url}
\usepackage{hyperref}

\usepackage[dvipsnames]{xcolor}

\usepackage{color}
\usepackage{pgf, tikz, pgfplots}
\usepackage{tcolorbox}

\usepackage{ amssymb, amsthm}
\usepackage{mathrsfs}
\input{mySymbol.sty}

\newtheorem{proposition}{\hspace{0pt}\bf Proposition}

\newtheorem{theorem}{\hspace{0pt}\bf Theorem}
\newtheorem{corollary}{\hspace{0pt}\bf Corollary}

\newtheorem{remark}{\hspace{0pt}\bf Remark}

\def\vec{\operatorname{vec}}

\begin{document}

\title{A Generalization of the Convolution Theorem and its Connections to Non-Stationarity and the Graph Frequency Domain}

\author{Alberto Natali,~\IEEEmembership{Student Member,~IEEE},  Geert Leus,~\IEEEmembership{Fellow,~IEEE}
\thanks{The authors are with the Faculty of Electrical Engineering, Mathematics and Computer Science, Delft University of Technology, 2628 Delft, The Netherlands. Emails: \{a.natali, g.j.t.leus\}@tudelft.nl. This work is part of the GraSPA project (project 19497 within the TTW OTP programme), which is financed by the Netherlands Organization for Scientific Research (NWO).}
\thanks{A preliminary version of this work appeared in~\cite{natali2022general}.}
}

\markboth{Journal of \LaTeX\ Class Files,~Vol.~14, No.~8, August~2021}%
{Shell \MakeLowercase{\textit{et al.}}: A Sample Article Using IEEEtran.cls for IEEE Journals}


\maketitle
\begin{abstract}
In this paper, we present a novel convolution theorem which encompasses the well known convolution theorem in (graph) signal processing as well as the one related to time-varying filters. Specifically, we show how a node-wise convolution for signals supported on a graph  can be expressed as another node-wise convolution in a frequency domain graph, different from the original graph. This is achieved through a parameterization of the filter coefficients following a basis expansion model. After showing how the presented theorem is consistent with the already existing body of literature, we discuss its implications in terms of non-stationarity. Finally, we propose a data-driven algorithm based on subspace fitting to learn the frequency domain graph, which is then corroborated by experimental results on synthetic and real data. 
\end{abstract}



\section{Introduction}
\label{sec:introduction}

 \IEEEPARstart{C}{onvolution} is the core operation in signal processing and machine learning systems. Its use is at the heart of digital 
filters~\cite{oppenheim2001discrete} for audio applications and time series prediction~\cite{lutkepohl2013introduction}, as well as for  convolutional neural networks in deep learning~\cite{lecun2015deep}, enabling scalable architectures and endowing a notion of locality among samples, properties exploited in, e.g., object recognition\cite{lecun1995convolutional}.

 \sn 
 The three key operations defining a convolution are the \textit{shift}, the \textit{scale} and the \textit{sum}. The shift is responsible to capture the underlying signal domain and brings the notion of locality and proximity among samples: in time, for instance, successive applications of the shift (in that case corresponding to a delay) give previous time samples. The scale defines how the shifted samples are weighted before summing them, and different weighting schemes lead to different structural properties (such as invariants) of the architecture implementing the convolution. The notion of regularity in time and in space, which are two very well structured domains, is reflected in the definition of their frequency domain. Specifically, a signal in these domains can be decomposed into elementary building blocks (such as sine waves) which endow a physical interpretation with a well understood meaning of variability. In a less structured domain modeled by a graph, this definition is not tight and multiple interpretations are possible. 

\sn
 Graph signal processing (GSP)~\cite{shuman2013emerging} extends the convolution principle to data residing on graphs by means of graph filters (GFs)~\cite{segarra2017optimal, shuman2018distributed}, architectures which are parametric on the mathematical structure defining the \textit{shift} operation. While in temporal signal processing this shift operator is mathematically represented by the (lower) shift/delay matrix, in GSP the \textit{graph shift operator} (GSO) depends on the underlying network domain. The eigendecomposition of the GSO reveals its respective frequencies: specifically, the eigenvalues of the shift are the frequencies. In temporal signal processing these frequencies are the well-known complex roots of unity which obey a natural ordering, and the associated (normalized) eigenvectors are the Fourier modes (remember that the delay matrix is a particular case of a circulant matrix). In GSP, however, different shifts have different  eigenvalues and hence different frequencies. In this case, an ordering purely based on their numeric value might not be meaningful and a more structured domain, for instance captured by a graph, might convey more information. This is the recent line of work explored in~\cite{leus2021dual,shi2022graph}, which we exploit here.


\sn
By relying on the novel notion of \textit{dual graph}~\cite{leus2021dual}, which models the support of the frequency domain as a graph, in this work we introduce a new convolution theorem which generalizes the (graph) convolution theorem and the one related to time-varying filters. To do this, we adopt so-called node-varying graph filters (NV-GFs)~\cite{segarra2017optimal} and we show how a node-varying convolution in one domain (captured by the \textit{primal} graph) can be expressed as a node-variant convolution in the other domain 
 (captured by the \textit{dual} graph), while remaining consistent with the pre-existing body of literature. Based on the proposed theorem, we outline models for non-stationary graph signals. Finally, we propose an algorithmic approach to learn the dual graph from data with a subspace fitting approach~\cite{viberg1991sensor} resorting to sequential convex programming~\cite{boyd2008sequential} to tackle the non-convexity of the problem. The validity of this approach is finally corroborated by simulations on synthetic and real data.

\subsection{Related Works}
 Although this work is novel in its genre, many of the concepts it relies on have been recently introduced. Modeling the frequency domain through a graph has been proposed in~\cite{leus2021dual, shi2022graph}, and similar efforts to interpret such a domain differently from ordering the eigenvalues of the GSO is given by means of embeddings in~\cite{saito2018can, cloninger2021natural}. Node-variant graph filters as a way to extend classical time-varying filters to the graph setting have been introduced in~\cite{segarra2017optimal}, while graph signal  stationarity arguments and their implications have been studied in~\cite{marques2017stationary, perraudin2017stationary}. Non-stationarity of random graph signals has been exploited in~\cite{shafipour2021identifying} in the context of network topology identification.

 \vspace{-.1cm}
\subsection{Contributions}
We summarize the specific contributions of this work as follows:
\begin{enumerate}
    \item We propose a convolution theorem which encompasses the (graph) convolution theorem and the one related to time-varying filters; we show how the latter are specific instantiations of the proposed theorem for particular choices of the GSO and the scaling scheme;
    \item We introduce novel non-stationary graph signal models;
    \item We devise an algorithmic procedure to learn the dual graph from input-output data. The problem formulation can be casted as a blind polynomial regression, as such also applicable to graph-agnostic tasks, such as polynomial interpolations and jitter correction in communication applications. The solution approach relies on a subspace fitting method and it is accompanied by a theoretical study of the ambiguities present in the problem.
    \item We showcase the validity of our findings on synthetic and real data with numerical simulations. 
\end{enumerate}

\section{Preliminaries}
\label{sec:prelims}



\textbf{Graphs and Graph Signals.} We consider data residing on a non-Euclidean domain, which we formally model by a graph $\ccalG=\left(\ccalV,\ccalE, \bbS \right)$, where $\ccalV=\{1, \ldots, N\}$ is the set of nodes (or vertices),  $\ccalE \subseteq \mathcal{V} \times \mathcal{V}$ is the set of edges, and $\bbS$ is an $N \times N$ matrix that represents the graph structure. The matrix $\bbS$ is called the graph shift operator (GSO), since it plays a role akin to the delay operator in temporal signal processing. Specifically, its entries $\left[\bbS\right]_{ij} \in \complex$ for $i \!\neq \! j$ are different from zero only if nodes $i$ and $j$ are connected by an edge;  typical examples of such a  matrix are the (weighted) adjacency matrix $\bbW$~\cite{sandryhaila2013discrete} and the graph Laplacian $\bbL$~\cite{shuman2013emerging}. A graph signal is then the vector $\bbx \in \complex^N$, where $x_i : \ccalV \to \complex$ is the value collected at node $i$.
 
 In this manuscript, for the sake of simplicity, we consider the shift operator $\bbS$ to admit an eigenvalue decomposition (EVD) written as $\bbS=\bbV\bbLambda\bbV^{-1} $, with $\bbV$ an invertible matrix collecting the eigenvectors and $\bbLambda= \Diag(\bblambda)$ a diagonal matrix collecting the eigenvalues $\bblambda$ of $\bbS$. A fundamental assumption in GSP is that the matrix $\bbV$ provides a basis for expressing signals living on $\bbS$, and with favorable discrete Fourier transform (DFT)-like properties providing a notion of frequency  similar to the one in temporal signal processing. For this reason, the matrix $\bbV^{-1}$ is often referred to as the graph Fourier transform (GFT) and the projection of $\bbx$ onto this basis, i.e.,  $\hat{\bbx}= \bbV^{-1} \bbx$ as the GFT signal.

\sn
\textbf{Filtering on Graphs.}
Given a graph $\bbS$, a classical graph filter (C-GF) of order $L-1$ is the matrix polynomial:
\begin{align}
\label{eq:node-invariant}
    \bbH(\bbp, \bbS)= \sum_{l=0}^{L-1} p_l \bbS^l,
\end{align}
where $\bbp= [p_0, \ldots, p_{L-1}]^\top \in \complex^L$ collects  the graph filter coefficients (taps). The application of the filter $\bbH(\bbp, \bbS)$ on a signal $\bbx$ to obtain a new signal $\bby$, i.e., $\bby=\bbH(\bbp, \bbS)\bbx$, is often referred to as graph filtering or \textit{graph convolution}, as it respects the scale-sum-shift principle of convolution. With a few simple calculations, it is  easy to show that in the (graph) frequency domain, a graph convolution is expressed as a pointwise multiplication; this is the (graph) convolution theorem, which can be expressed as follows:
\begin{tcolorbox}
\vspace{-.4cm}
\begin{align}
\label{eq:graph-convolution-theorem}
    \bby &= \sum_{l=0}^{L-1} p_l\bbS^l \bbx  \quad  \quad  \hat{\bby}=\sum_{l=0}^{L-1} p_l\bbLambda^l \hat{\bbx}
 \end{align}
\end{tcolorbox}
\noindent with $\hat{\bby}=\bbV^{-1}\bby$ the GFT of $\bby$. Notice that such filter is isotropic, meaning that for each $l=0, \ldots, L-1$,  the filter coefficient $p_l$ is shared among all the nodes of the shifted signal $\bbS^l\bbx$; for this reason a C-GF is an example of a \textit{node-invariant} graph filter.

\sn
A more versatile and flexible version  of~\eqref{eq:node-invariant} is the so-called  \textit{node-variant} graph filter~\cite{segarra2017optimal}, which allows a per-node weighting scheme of each shifted version of the input signal. Due to its relevance in this work, we distinguish among two flavours of a NV-GF, henceforth referred to as \textit{type-I} and \textit{type-II}, defined, for a given a graph $\bbS$ and fixed order $L-1$, respectively as:
\vspace{-.1cm}
\begin{align}
\label{eq:node-variant}
    \bbH_I(\bbP, \bbS) &= \sum_{l=0}^{L-1} \Diag(\bbp_l) \bbS^l, \\
    \label{eq:node-variant-2}
    \bbH_{II}(\bbP, \bbS) &= \sum_{l=0}^{L-1}  \bbS^l \Diag(\bbp_l),
\end{align}
where $\bbP \in \complex^{N \times L}$ collects the filter coefficients $\bbP=[\bbp_0, \ldots, \bbp_{L-1}]$ with $\bbp_l:= [p_{l1}, \ldots, p_{lN}]^\top \in \complex^N$ the $l$-th hop filter tap vector. As a short-hand notation, we will use   $\bbH_I$ and $\bbH_{II}$ to refer to the NV-GF in~\eqref{eq:node-variant} and~\eqref{eq:node-variant-2}, respectively; when convenient for clarity of exposition, we will explicitly write  $\bbH_I(\bbP,\bbS)$ or $\bbH_{II}(\bbP,\bbS)$ concordantly.  The application of  a NV-GF on a signal $\bbx$ to obtain a new signal $\bby$ will be referred to as \textit{node-variant graph convolution}. 
From a theoretical point of view, both NV-GF types have the same expressive behavior, yet the order of shifting and weighing is reversed. Specifically, in type-I, each node performs a linear combination of the (shifted) signal values of neighboring nodes, where the weights of the linear combination are neighbor-specific; in type-II, each node performs a linear combination of the (shifted) signal values of neighboring nodes, which have been already scaled by such nodes. Nonetheless, both can be implemented with the same complexity and in a distributed manner~\cite{segarra2017optimal}. 

\smallskip\noindent\textbf{Dual Graph.} Although often not explicitly stated in the academic literature, the support of the GFT signal $\hat{\bbx}$ is described by the eigenvalues $\bblambda$ of $\bbS$, which correspond to a discretization/sampling of a continuous domain, either  the real line $\reals$ or the complex plane $\complex$. This is consistent with the discrete signal processing notion of frequency domain: when $\bbS$ represents a cycle graph, possibly capturing the time domain, its eigenvector matrix $\bbV^{-1}$ coincides with the (normalized) DFT matrix, and its eigenvalues $\bblambda$ with the complex frequencies on the unit circle, i.e.,  $\bblambda= [1, e^{-j2\pi/N}, \ldots, e^{-j2\pi(N-1)/N}]$.
However, a modern line of research attempts to model the (graph) frequency domain by means of a graph~\cite{leus2021dual}. The motivation behind this line of research relies on the fact that classical signal processing tasks usually performed in the frequency domain, such as frequency-shifting, do not have their counterpart in GSP. Furthermore, given that a graph signal is inherently associated with a graph structure, it is desirable to establish a corresponding Fourier representation that is also inherently linked to a graph structure. This leads to the notion of a dual graph\footnote{Not to be confused with the dual graph notion in graph theory, as the graph which has a vertex for each face of the original graph.} $\bbS_f= \bbV_f \bbLambda_f \bbV_f^{-1}$, which represents the support for the GFT signal $\hat{\bbx}$. Because we want the GFT $\bbV_f^{-1}$ associated to the dual graph to map $\hat{\bbx}$ back to the signal $\bbx$, i.e.,  $\bbx = \bbV_f^{-1} \hat{\bbx}$, and we know that $\bbx=\bbV\hat{\bbx}$, we must have $\bbV_f=\bbV^{-1}$.

Thus, the primal graph provides \textit{spectral templates} for the graph frequency domain, i.e., the eigenvectors $\bbV_f$ for the dual graph $\bbS_f$ are known once we know those of $\bbS$. The only unknown is then the eigenvalue matrix $\bbLambda_f:= \Diag(\bblambda_f)$, which can be found, for instance, with an axiomatic or an optimization approach~\cite{leus2021dual}, or in a data-driven manner as we will show in Section~\ref{sec:learning}. Although~\cite{shi2022graph} proposes $\bblambda_f = \bblambda^\star$, we do not see this as a favorable definition, since in our view it only holds when specified to the ``temporal'' graph; in all the other cases, especially for undirected graphs, it would implies that primal and dual eigenvalues always coincide. 
Such an interpretation would be inconsistent with the desirable properties highlighted in~\cite[Axioms (1-3)]{leus2021dual}.


\section{An Encompassing Convolution Theorem}
\label{sec:theorem} 


\vspace{-.3cm}
In this section we propose a convolution theorem which encompasses the graph convolution theorem~[cf.~\eqref{eq:graph-convolution-theorem}] introduced in Section~\ref{sec:prelims} and the convolution theorem related to time-varying filters, which will be introduced later on to highlight similarities and differences. This generalization is made possible by using the node-variant graph filters~\eqref{eq:node-variant} and \eqref{eq:node-variant-2}   with an appropriate parametrization of the filter coefficients. Specifically, we show how a limited order NV-GF in the primal domain can be expressed as a limited order NV-GF in the dual domain. This is formally stated in the following theorem.

\begin{theorem}[Node-variant convolution theorem]
\label{theorem:one}
\textit{ Consider a type-I NV-GF $\bbH_I$ defined over the graph $\bbS$ with filter taps $\{\bbp_l\}_{l=0}^{L-1}$, i.e.,  $\bbH_I(\bbP,\bbS)$, and assume that a dual graph $\bbS_f$ with dual graph frequencies $\bblambda_f$ describing the dual domain is given. Assume also that each filter tap vector $\bbp_l$ can be expressed as a polynomial of order $K-1$ in $\bblambda_f$. Then, there exists a set of coefficients $\{\hat{\bbp}_k\}_{k=0}^{K-1}$ for which the type-I NV-GF $\bbH_I(\bbP,\bbS)$ in the primal domain corresponds to a type-II NV-GF $\bbH_{II}$ on the dual graph $\bbS_f$ with filter taps $\{\hat{\bbp}_k\}_{k=0}^{K-1}$, i.e., $\bbH_{II}(\hat{\bbP},\bbS_f)$.}
\end{theorem}

\begin{proof}
By multiplying both sides of~\eqref{eq:node-variant} with the GFT matrix $\bbV^{-1}$, we have:
\begin{align}
\label{eq:gft-node-variant}
    \hat{\bby}&= \bbV^{-1} \sum_{l=0}^{L-1} \Diag(\bbp_l) \bbS^l \bbx =  \bbV^{-1}  \sum_{l=0}^{L-1} \Diag(\bbp_l) \bbV \bbLambda^l  \hat{\bbx}.
\end{align}
Next, we use a basis expansion model (BEM)~\cite{giannakis1998basis} to express the NV filter coefficients $\{\bbp_l\}_{l=0}^{L-1}$ as a linear combination of a dual graph dependent basis. Specifically, we express each $\bbp_l$ through powers of the dual eigenvalues $\bblambda_f$, representing our basis expansion; that is:
\begin{align}
\label{eq:primal-coefficients}
    \bbp_l= \sum_{k=0}^{K-1} c_{lk}\bblambda_f^k= \bbPsi_f \bbc_l
\end{align}
with  $\bbPsi_f:=[\boldsymbol{1}, \; \bblambda_f, \ldots,  \bblambda_f^{K-1}]$ the Vandermonde matrix of dual eigenvalues and $\bbc_l:= [c_{l0}, \ldots, c_{l(K-1)}]^\top$ the expansion coefficients for the $l$-th primal filter tap vector $\bbp_l$. With this choice, substituting~\eqref{eq:primal-coefficients} in~\eqref{eq:gft-node-variant}, we have:%
\begin{align}
\label{eq:dual-node-variant}
    \hat{\bby}&= \nonumber \bbV^{-1}  \sum_{l=0}^{L-1} \Diag(\sum_{k=0}^{K-1} c_{lk}\bblambda_f^k) \bbV \bbLambda^l  \hat{\bbx} \\ \nonumber
    &=  \sum_{k=0}^{K-1} \bbV^{-1}    \Diag(\bblambda_f^k) \bbV \Diag(\sum_{l=0}^{L-1} c_{lk} \bblambda^l)  \hat{\bbx}\\ 
    &= \sum_{k=0}^{K-1} \bbS_f^k \Diag(\hat{\bbp}_k)\hat{\bbx}
\end{align}
where $\hat{\bbp}_k:=\sum_{l=0}^{L-1} c_{lk} \bblambda^l= \bbPsi \hat{\bbc}_k$ is the $k$-th hop filter tap vector on the \textit{dual} graph, with $\bbPsi:= [\boldsymbol{1}, \bblambda, \ldots, \bblambda^{L-1}]$ the Vandermonde matrix of primal eigenvalues, and $ \hat{\bbc}_k:=[c_{0k}, \ldots, c_{(L-1)k}]^\top$ the expansion coefficients for the $k$-th dual filter tap vector $\hat{\bbp}_k$. So in the frequency domain, we also obtain a NV-GF denoted as $\bbH_{II}= \sum_{k=0}^{K-1} \bbS_f^k \diag(\hat{\bbp}_k)$ . Whenever the dependency on the dual filter coefficients and shift operator is necessary, we use $\bbH_{II}(\hat{\bbP},\bbS_f)$, where  $\hat{\bbP}$ is the $N \times K$ matrix of coefficients $\hat{\bbP}=[\hat{\bbp}_0, \ldots, \hat{\bbp}_{K-1}]$. 
\end{proof}

\noindent This theorem allows us to delineate a general convolution theorem encompassing~\eqref{eq:graph-convolution-theorem} as a special case, which relies on node-variant graph filtering and pictorially described in Fig.~\ref{fig:duality}, as follows:
\begin{tcolorbox}
\vspace{-.2cm}
  \begin{align}
\label{eq:duality-node-varying}
   \hspace{-.2cm}\bby &\!=\! \sum_{l=0}^{L-1} \Diag(\bbp_l)\bbS^l \bbx  \quad  \quad  \hat{\bby}\!=\!\sum_{k=0}^{K-1} \bbS_f^k \Diag( \hat{\bbp}_k) \hat{\bbx} \\
   \label{eq:duality-coefficients}
    &\bbp_l=  \bbPsi_f \bbc_l \quad   \quad  \quad    \quad \quad    \quad \hat{\bbp}_k=  \bbPsi  \hat{\bbc}_k
\end{align}
\end{tcolorbox}
The connection between the primal and the dual node-variant graph filters defined in~\eqref{eq:duality-node-varying} is given by the $K \times L$  expansion coefficients conveniently stored in the matrix $\bbC=[\bbc_0, \ldots,\bbc_{L-1}]= [ \hat{\bbc}_0, \ldots,  \hat{\bbc}_{K-1}]^\top$. This enables also to concisely express the node-variant coefficients in the primal and dual domain as $\bbP= \bbPsi_f \bbC$ and $\hat{\bbP}= \bbPsi\bbC^\top$, respectively.

 \begin{figure}[t!]
     \centering
  \includegraphics[width=.95\columnwidth, trim= 5cm 4cm 5cm 3cm, clip=True]{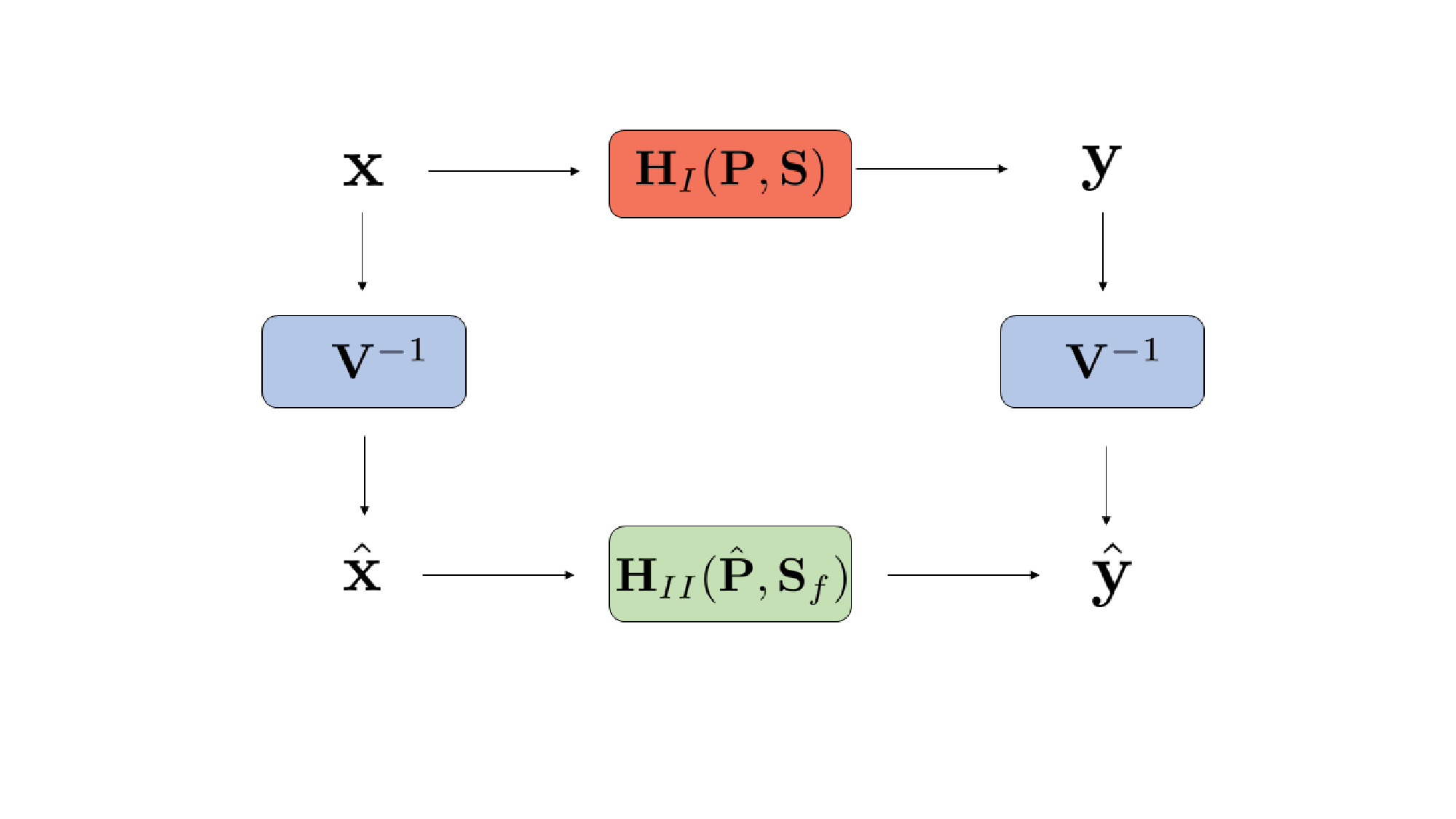}
  \vspace{-.25cm}
     \caption{General convolution theorem. A node-variant graph convolution in the primal domain is equivalent to a node-variant graph convolution in the dual domain.}
     \label{fig:duality}
 \end{figure}

\begin{remark}
    The same type of theorem construction can be obtained by reversing the types of filters adopted in the primal and dual domain; that is, applying a type-II NV-GF in the primal domain is equivalent to applying a type-I NV-GF in the dual domain, with the graph filter coefficients following the parametrization in~\eqref{eq:duality-coefficients}.                                                               
\end{remark}

 %
 \begin{corollary}
 \label{corollary: categoric}
 Given a graph signal $\bbx$, the application of a  node-variant graph filter $\bbH_I(\bbP,\bbS)$ in the primal domain followed by the GFT $\bbV^{-1}$ is equivalent to the application of the GFT followed by a node-variant graph filter $\bbH_{II}(\hat{\bbP},\bbS_f)$ in the dual domain. In other words, it holds (see also Fig.~\ref{fig:duality}):
\begin{align}
\label{eq: categoric}
   \bbV^{-1} \bbH_I(\bbP,\bbS) =  \bbH_{II}(\hat{\bbP},\bbS_f)\bbV^{-1} .
\end{align}
 %

 \end{corollary}

 In temporal signal processing, the frequency representation of windowing in the time domain is the convolution between the spectra of the signal and the window.  Because a node-variant convolution of order $L-1$ is nothing else than the application of $L$ windows on shifted versions of the input graph signal $\bbx$, a similar result can be derived in the graph setting; the following corollary expresses this.
 \begin{corollary}
\label{corollary:duality}
 Given an input graph signal $\bbx$ and a type-I NV-GF $\bbH_I(\bbP,\bbS)$, with each $\{\bbp_l\}_{l=0}^{L-1}$ parametrized as in~\eqref{eq:duality-coefficients}, a node-variant graph convolution of order $L-1$ in the primal domain is equivalent to the sum of $L$ classical graph convolutions in the dual domain, each one with as input a (modulated) version of $\hat{\bbx}$; that is:
  \begin{align}
     \hat{\bby}&= \sum_{l=0}^{L-1} \bbH(\bbc_l, \bbS_f) (\bblambda^l \odot \hat{\bbx})\\
     &= \bbH(\bbc_0,\bbS_f)\hat{\bbx} + \ldots + \bbH(\bbc_{L-1}, \bbS_f) (\bblambda^{L-1} \odot \hat{\bbx})
 \end{align}
 \end{corollary}
 
 \begin{proof}
 From the first equality of~\eqref{eq:dual-node-variant}, we have:
 \begin{align}
     \hat{\bby}&=\sum_{l=0}^{L-1} \sum_{k=0}^{K-1} c_{lk} \bbV^{-1} \Diag(\bblambda_f^k) \bbV \bbLambda^l \hat{\bbx} \nonumber\\ \nonumber
      &= \sum_{l=0}^{L-1} \left(\sum_{k=0}^{K-1} c_{lk} \bbS_f^k\right) \bbLambda^l \hat{{\bbx}}\\
      \label{eq:modulation}
      &=\sum_{l=0}^{L-1} \bbH(\bbc_l, \bbS_f) (\bblambda^l \odot \hat{\bbx}) 
 \end{align}
 Notice how the filter coefficients in~\eqref{eq:modulation} are the expansion coefficients $\bbc_l$ associated to the primal filter coefficients $\bbp_l$.
 \end{proof}
\noindent
An important consequence of Corollary~\ref{corollary:duality} is that windowing in the primal domain is equivalent to a C-GF in the dual domain; we will rely upon this when introducing non-stationary signal models in Section~\ref{sec:non-stationary}.
 
 \subsection{Consistency with the graph convolution theorem}
\label{remark:constant}
Because a C-GF is a NV-GF with constant filter taps, we expect that our introduced theory  encompasses the existing one. Indeed, we  can formally show that the graph convolution theorem~\eqref{eq:graph-convolution-theorem}  falls within the introduced theory. To see this, consider $\bbp_l= p_l\boldsymbol{1}, \; \forall \; l$, i.e., the case in which each vector of filter taps $\bbp_l$ is constant over the nodes, thus corresponding to the C-GF as in~\eqref{eq:node-invariant}. The column space of $\bbP$ is then one-dimensional, specifically spanned by the all-one vector. By construction, the first column of the Vandermonde matrix is the all-one vector. Thus any $\bblambda_f$ can be used as it will be zeroed-out by the matrix $\bbC$. As we will see later on, this also means that we cannot learn any dual graph from stationary graph signals, since any $\bblambda_f$ would suffice.  Specifically,  from~\eqref{eq:duality-coefficients}, we have that $\bbc_l$ necessarily needs to be $\bbc_l= [p_l, \bb0^\top]^\top$, and overall:
\begin{align}
    [p_1 \boldsymbol{1}, \ldots, p_{L-1} \boldsymbol{1}]\!=\!
    \begin{bmatrix}
    1 & \cdots & \lambda_{0,f}^{K-1}\\ 
    \vdots & \ddots & \vdots \\
    1 & \cdots & \lambda_{f,N-1}^{(K-1)}
    \end{bmatrix}
    \begin{bmatrix}
    p_1 & \cdots & p_{L-1}\\ 
    0 & \cdots & 0\\
    \vdots & \ddots & \vdots \\
    0 & \cdots & 0
    \end{bmatrix}
\end{align}
meaning that only the first row $\hat{\bbc}_0$ of $\bbC$ is different from zero. In particular, from the right equation in~\eqref{eq:duality-coefficients} this implies that:
\begin{align}
    \hat{\bby}=  \bbS_f^0 \Diag( \hat{\bbp}_0) \hat{\bbx} = \Diag(\bbPsi \bbp) \hat{\bbx}.
\end{align}
 This shows how the proposed theory fits within the  principle that a classical graph convolution (node-invariant GF) is a pointwise multiplication in the frequency domain.
 
 \sn
 Likewise, one can similarly show that if $\hat{\bbp}_k= \hat{p}_k \boldsymbol{1}, \; \forall \; k $, i.e., the case in which each vector of filter taps $\hat{\bbp}_k$ in the frequency domain is constant over the nodes (frequencies), then only $\bbc_0=\hat{\bbp}:=[\hat{p}_0, \ldots, \hat{p}_{K-1}]^\top$ is different from zero. This leads to a pointwise multiplication (windowing) [cf.~\eqref{eq:duality-coefficients} left] in the primal domain:
 \begin{align}
     \bby= \Diag(\bbp_0)\bbS^0\bbx= \Diag(\bbPsi_f\hat{\bbp})\bbx,
 \end{align}
which also complies with Corollary~\ref{corollary:duality}.

 \sn
 Another interesting relation arises when considering a NV-GF with $p_{li} = p_i$ for all $l$; i.e.,  the case in which each node $i$ uses the same weight $p_i$ (possibly different from $p_j$ of node $j$) for the diffused sequence $\{\bbx, \bbS \bbx, \ldots, \bbS^{L-1}\bbx\}$. In this case, multiple $\bblambda_f$ and $\bbC$ satisfy the theorem; for instance we can choose $\bblambda_f = [p_1, \ldots, p_N]^\top$ and  $\bbc_l= [0, 1, \bb0^\top]^\top$, i.e.,:
\begin{align}
\begin{bmatrix}
    p_1 \boldsymbol{1}^\top \\
       \vdots \\
    p_N \boldsymbol{1}^\top
\end{bmatrix} = 
    \begin{bmatrix}
    1 & p_1 & \cdots & p_1^{K-1}\\ 
    \vdots & \vdots & \vdots \\
    1 & p_N & \cdots & p_N^{(K-1)}
    \end{bmatrix}
     \begin{bmatrix}
    0 & \cdots & 0\\ 
    1 & \cdots & 1\\
    \vdots & \cdots & \vdots \\
    0 & \cdots & 0
    \end{bmatrix}
\end{align}

 This implies that:
 \begin{align}
    \hat{\bby}=  \bbS_f \Diag( \hat{\bbp}_1) \hat{\bbx} = \bbS_f \Diag(\bbPsi \boldsymbol{1}) \hat{\bbx},
\end{align}
so that each $\hat{x}_i$ is multiplied with $\hat{p}_{1i}= [1, \lambda_i, \ldots, \lambda_i^{L-1}] \boldsymbol{1}$.

 
 \subsection{Relationship with time-varying channel propagation}
 The proposed theory also generalizes, to the graph setting, concepts which are familiar in the context of time-varying channel propagation~\cite{tang2007pilot}, arising for instance in mobile communication scenarios. In that case, the received signal $y$ at time $n$, i.e., $y[n]$\footnote{We use square brackets to indicate that the argument is a time index and not a graph node.}, is modeled as:
 \begin{align}
 \label{eq:time-varying-channel}
     y[n]= \sum_{l=0}^{L-1} p[n,l] x[n-l],
 \end{align}
 where $p[n,l]$ denotes the channel impulse response of the $l$-th path at the $n$-th time instant, and $x[n-l]$ is the transmitted signal at the $(n-l)$-th time instant. The gains  associated to the different paths are assumed to be time-varying and approximated by a basis expansion model~\cite{giannakis1998basis}; specifically:
 \begin{align}
 \label{eq:bem}
     \bbp_l= \sum_{k=0}^{K-1} c_{lk} \bbb_k,
 \end{align}
 where $ \bbp_l=[p[0,l], \ldots, p[N-1,l]]^\top$ stores the evolution of the filter impulse response over the $N$ time instants, $\bbb_k \in \reals^N$ is the $k$-th basis function, and $c_{lk}$ is the coefficient associated to the $l$-th path and the $k$-th basis function. This alleviates the effort of having to deal with  $NL$ channel coefficients (usually a very high number), and converts the  model into a simpler one with only $LK$ BEM coefficients.
 
 It is easy to show that we can write~\eqref{eq:time-varying-channel} in matrix-vector form, by taking into account~\eqref{eq:bem}, as:
 \begin{align}
 \label{eq: time-varying-vector}
    \bby=  \sum_{k=0}^{K-1} \Diag(\bbb_k) \left(\sum_{l=0}^{L-1}c_{lk}\bbD^l\right) \bbx
 \end{align}
 where $\bbx= [x[0], \ldots, x[N-1]]^\top$ and $\bbD$ is the $N \times N$ lower delay matrix; notice how the matrix $\sum_{l=0}^{L-1}c_{lk}\bbD^l$ implements a standard convolutional filter in time and observe its similarities with the left equation in~\eqref{eq:duality-node-varying}. Next, denote with $\bbF \in \mathbb{C}^{N \times N}$ the normalized DFT matrix, and with $\bbf_k$ its $k$th column; the classical complex exponential BEM uses the Fourier basis as basis functions in~\eqref{eq:bem}, i.e., $\bbb_k=\sqrt{N} \bbf_k$. As such the gains $\bbp_l$ associated to the $l$-th path [cf.~\eqref{eq:bem}] are modelled as smoothly-varying over time and hence expressed with a small number $K$ of Fourier basis. Increasing $K$ accommodates for faster changes.

\sn
With this choice, \eqref{eq: time-varying-vector} can be expressed in the frequency domain as:
 \begin{align}
 \label{eq:time-varying-convolution-theorem}
   \nonumber  \hat{\bby}&= \bbF\bby=  \sum_{k=0}^{K-1} \bbF\Diag(\sqrt{N}\bbf_k)\sum_{l=0}^{L-1} c_{lk} {\bf F}^H {\rm Diag}(\sqrt{N}{\bf f}_l) {\bf F}\bbx \\
     &= \sum_{l=0}^{L-1} \left(  \sum_{k=0}^{K-1} c_{lk} {\bf D}^k \right) {\rm Diag}(\sqrt{N}{\bf f}_l)\hat{\bbx}.
 \end{align}
 While in~\eqref{eq: time-varying-vector} the matrix $\bbD$ shifts in the time domain, in~\eqref{eq:time-varying-convolution-theorem} it shifts in the frequency domain; however, such shift matrix is the same in both domains. This is different when looking at the graph counterpart in~\eqref{eq:dual-node-variant}, where the two shift matrices might be different. 
 
 \begin{remark}
      It is worthwhile to point out that the notion of smoothness for signals and filter coefficients is different in temporal and graph signal processing. In the time domain the basis to express smoothness for signals and filter coefficients is the same, both coinciding with the normalized DFT matrix. In GSP, the basis to express smoothness of graph signals and graph filter coefficients is different. Our theory shows that smoothness of graph signals is determined by the eigenvectors of the primal graph, while smoothness of the filter coefficients is determined by the Vandermonde matrix containing the dual graph frequencies.
 \end{remark}

 All in all~\eqref{eq:time-varying-convolution-theorem} is the time domain counterpart of~\eqref{eq:dual-node-variant}, by choosing the primal eigenvector matrix $\bbV$ to be $\bbV=\bbF^H$ and the basis functions $\bblambda_f^k$ to be $\bblambda_f^k= \sqrt{N}\bbf_k$.
 \sn

 %
 %


 \subsection{Non-stationarity}
 \label{sec:non-stationary}
In this section we study how and where  
(non-)stationarity of graph signals stands within the introduced theory. From~\cite{marques2017stationary}, a process $\bby$ is said to be weakly stationary on a GSO $\bbS$ if the covariance matrix $\bbC_y:= \mathbb{E}[\bby\bby^H]$ commutes with $\bbS$ or, equivalently, if $\bby$ can be written as the output of a C-GF $\bbH$ [cf.~\eqref{eq:node-invariant}] when excited with a white input $\bbx$, i.e., $\bby=\bbH\bbx$. As a consequence, the covariance matrix $\bbC_{\hat{y}}$ of the GFT process $\hat{\bby}$ is diagonal, revealing the power spectral density of the process $\bby$ on its diagonal.

\sn
While~\cite{marques2017stationary} offers conditions to identify and model stationary graph signals, it does not explore non-stationary signal models; this is our first attempt in that direction. Worth to mention is the work of~\cite{shafipour2021identifying}, where a network topology identification approach is put forth to learn the GSO $\bbS$ given a set of realizations of a non-stationary graph signal $\bby$ modelled with a classical graph convolution~\eqref{eq:graph-convolution-theorem}. There, however, non-stationarity is only considered with respect to node-invariant GFs, thus restricting the non-stationarity model taxonomy. Indeed, the (non-)stationarity of a random graph signal $\bby$ obtained as $\bby=\bbH\bbx$ for a generic graph filter $\bbH$ and an excitation input $\bbx$, depends either on the type of graph filter or the properties of the input $\bbx$. Precisely a non-stationary graph signal $\bby$ can be either modelled as the output of a shift-invariant graph filter with a non-stationary input or as the output of a node-variant graph filter when excited with a white input. The following property and propositions formally describe these claims.

\sn
\textbf{Property 1.}
Given input $\bbx$ and output $\bby=\bbH_{I}(\bbP,\bbS)\bbx$, it holds:
\begin{align}
    \bbC_y= \bbH_{I}(\bbP,\bbS)\bbC_x\bbH_{I}(\bbP,\bbS)^H
\end{align}
with $\bbC_x:= \e[\bbx\bbx^H]$. Moreover:
\begin{align}
\label{eq:covariance-gft}
    \bbC_{\hat{y}}= \bbV^{-1}\bbC_y \bbV.
\end{align}

\noindent Depending on the structure of the graph filter $\bbH_I(\bbP, \bbS)$ and the input signal $\bbx$, the ensuing propositions can be derived. Unless explicitly stated differently, we assume that the graph filter coefficient matrix $\bbP$ respects the parametrization in~\eqref{eq:duality-coefficients}, i.e., $\bbP= \bbPsi_f \bbC$, for some order $L, K$.

\begin{proposition}
    If $\bbx$ is a white graph signal, then [cf. Corollary~\ref{corollary: categoric}]:
    \begin{align}
    \label{eq:Cy}
        \bbC_y &= \bbH_{I}(\bbP,\bbS)\bbH_{I}(\bbP,\bbS)^H \\
        \label{eq:C_yhat}
        \bbC_{\hat{y}} &= \bbH_{II}(\hat{\bbP},\bbS_f) \bbH_{II}(\hat{\bbP},\bbS_f)^H.
    \end{align}
     In general, this means that $\bby$ is non-stationary on $\bbS$ and $\hat{\bby}$ is non-stationary on $\bbS_f$.
\end{proposition}

\begin{proposition}
\label{prop:yhat-stationary-on-dual}
    If $L=1$ and $\bbx$ is a white signal (stationary by definition), then $\hat{\bby}$ is stationary on $\bbS_f$.
\end{proposition}

\noindent From~\eqref{eq:covariance-gft}, we can see that if $\bbC_y$ is diagonal, then $\bbS_f$ commutes with $\bbC_{\hat{y}}$ and as such $\hat{\bby}$ is stationary on $\bbS_f$, with a power spectral density (in the primal domain) equal to the eigenvalues of $\bbC_{\hat{y}}$. This can happen only if $L=1$ and $\bbx$ is a white signal, where the convolution simply becomes $\bby= \Diag(\bbp_0)\bbx$, i.e., a windowing in the primal domain. The covariance matrix $\bbC_y$ is then $\bbC_y= \Diag(\bbp_0)^2$ and the cross correlation in $\bby$ is zero. In this case the covariance matrix $\bbC_{\hat{y}}$ of the process $\hat{\bby}$ in the dual domain is not diagonal in general, since it reads as:
 \begin{align}
 \label{eq:cov-GFT}
      \bbC_{\hat{y}}= \bbV^{-1} \Diag(\bbp_0)^2 \bbV.
\end{align}
From~\eqref{eq:cov-GFT}, we can then conclude that $\bby$ is non-stationary on $\bbS$ and, more importantly, $\bbS_f \bbC_{\hat{y}} = \bbC_{\hat{y}} \bbS_f$, i.e., $\bbS_f$ commutes with $\bbC_{\hat{y}}$. This implies that $\hat{\bby}$ can be expressed as the output of a node-invariant graph filter in the dual domain when excited with white input, i.e., $\hat{\bby}= \bbH(\hat{\bbp}, \bbS_f) \hat{\bbx}$ for some limited order filter coefficients $\hat{\bbp}$, rendering the estimation of the dual graph (see Section~\ref{sec:learning}) a node-invariant graph filter estimation problem~\cite{segarra2017optimal, natali2020topology}.
This extends the classical notion of windowing in time domain (for instance used in power spectral density estimation), which corresponds to a frequency-invariant convolution in the frequency domain. A generalization of this is given for $L>1$ and a general signal $\bbx$, for which Corollary~\ref{corollary:duality} applies.

\begin{proposition}
    If $L=1$ and $\bbx$ is a \textit{non-white} yet stationary graph signal with covariance $\bbC_x=\bbV \bbLambda_x \bbV^H$ , then $\hat{\bby}$ is not stationary on $\bbS_f$ since:
\begin{align}
    \bbC_{\hat{y}}= \bbH(\bbc; \bbS_f) 
 \bbLambda_x \bbH(\bbc; \bbS_f)^H
\end{align}
for some coefficients $\bbc \in \reals^K$, does not commute with $\bbS_f$.
\end{proposition}

\begin{proposition}
    If $\bbp_l = p_l \boldsymbol{1}, \; \forall \; l$, and $\bbx$ is not stationary on $\bbS$, then $\bby$ is not stationary on $\bbS$ and $\hat{\bby}$ is not stationary on $\bbS_f$. Indeed we have that the covariance matrix of $\bby$:
\begin{align}
    \bbC_y = \bbH(\bbp; \bbS) \bbC_x \bbH(\bbp; \bbS)^H
\end{align}
 does not commute with $\bbS$; likewise $\bbS_f$ does not commute with $\bbC_{\hat{\bby}}= \bbV^{-1} \bbC_\bby \bbV$.
\end{proposition}

\sn In other words, a node-invariant graph convolution of a non-stationary process $\bbx$ results in a non-stationary process $\bby$ on $
\bbS$ (and non-stationary GFT $\hat{\bby}$ on $\bbS_f$). While this result is not novel, we include it here to ensure a comprehensive coverage.

\sn
The introduced propositions provide a way to artificially generate non-stationary graph signals on a given GSO $\bbS$ by filtering white noise, as explained next.

\sn
\textbf{Generating non-stationary graph signals.} 
 Remember that $\bbC_{\hat{\bby}}= \bbV^{-1} \bbC_{\bby} \bbV$ has to be diagonal for $\bby$ to be stationary on $\bbS$. Thus, non-stationary graph signals on $\bbS$ can be generated as long as $\bbC_{\bby}$ is not diagonalizable by $\bbV$ (which would render $\bbC_{\hat{\bby}}$ diagonal). In particular, if we want our GFT random process $\hat{\bby}$ to have a specific covariance matrix equal to $\bbC_{\hat{\bby}}= \bbV^{-1} \bbC_\bby \bbV$, for some positive semidefinite (PSD) $\bbC_\bby$ we can generate random samples $\hat{\bby}$ as:
\begin{align}
\label{eq:random_GFT}
    \hat{\bby}= \bbV^{-1} \bbR \bbV \hat{\bbx}
\end{align}
where $\bbR $ is a matrix such that $\bbC_y = \bbR \bbR^H$ and $\hat{\bbx}$ is a white input; equivalently, in the primal domain:
\begin{align}
\label{eq:random_y}
    \bby= \bbR \bbx.
\end{align}
A simple example of such generation process is given by setting $\bbC_\bby= \Diag(\bbp)$ for some $\bbp \succcurlyeq \bb0$, so that~\eqref{eq:random_y} is a windowing operation in the primal domain, corresponding to a node-invariant graph convolution [cf. \eqref{eq:random_GFT}] in the dual domain. With this choice, however,  $\hat{\bby}$ is stationary on $\bbS_f$ [cf. Proposition~\ref{prop:yhat-stationary-on-dual}]. If this is not sought-after, a  general  non diagonal PSD $\bbC_\bby$ should be used.

%
%


\section{Dual Graph Identification}
\label{sec:learning} 
So far we have assumed the knowledge of $\bbS_f$.
In this section, we put forth a data-driven procedure to learn the dual graph eigenvalues $\bblambda_f$ in such a way that the resulting graph $\bbS_f$ respects the theory developed in Section~\ref{sec:theorem}. For simplicity, we  restrict our attention to the case of an undirected primal graph with real-valued GSO $\bbS$ and real-valued graph signals and filter coefficients. The problem setting is the following: consider $T$ graph signals $\bbY=[\bby_1, \ldots, \bby_T]$ which can be modelled as non-stationary on the graph $\bbS$ as the result of filtering $T$ (possibly unknown) input graph signals $\bbX=[\bbx_1, \ldots, \bbx_T]$ with a NV-GF $\bbH_I(\bbP,\bbS)$, i.e., $\bbY=\bbH_I(\bbP,\bbS)\bbX$. In particular, similar to temporal signal processing, we assume that the orders $K$ and $L$ are much smaller than $N$. Then the problem of identifying the dual graph can be formalized as follows:

\sn
\begin{tcolorbox}
   \textit{Given $\bbY$ (and possibly $\bbX$), find a dual graph $\bbS_f$ which is consistent with Theorem~\ref{theorem:one}. In other words, find the dual eigenvalues $\bblambda_f$ such that a NV-GF $\bbH_I$ on $\bbS$ with $L \ll N$ can be expressed as a NV-GF $\bbH_{II}$ on $\bbS_f$ with $K \ll N$.}
\end{tcolorbox}

\sn 
To tackle this problem, we adopt a two-step approach: in  \textit{step i)} we learn the filter taps $\bbP$ of the NV-GF $\bbH_I(\bbP,\bbS)$ which best fit the available data, developing two distinct approaches for the input-output and output-only scenarios; in the second \textit{step ii)} we find the dual eigenvalues $\bblambda_f$ by exploiting~\eqref{eq:duality-coefficients}, i.e., we fit the model $\bbP = \bbPsi_f \bbC$, which is a specific structured matrix factorization with a Vandermonde factor. We solve this problem by following the recently proposed approach in~\cite{natali2023blind} relying on a subspace method followed by successive convex programming.

\subsection{Graph Filter Estimation with Input-Output Data}
\label{subsec:learn-P}
The goal of this section is to estimate the graph filter coefficients $\bbP$ from a set of data pairs $\ccalD= \{(\bbx_t, \bby_t) \vert \; t= 1, \ldots, T\}$, all obtained using the same node-varying graph filter. By vectorizing the expression $\bbY=\bbH_I(\bbP,\bbS)\bbX$ we obtain:
\begin{align}
\label{eq:vectorization-P-estimation}
    \bby&= \sum_{l=0}^{L-1} \nonumber\vec(\Diag(\bbp_l)\bbS^l\bbX)\\
    \nonumber &= \sum_{l=0}^{L-1}((\bbS^l\bbX)^\top \circ \bbI_N) \bbp_l\\
    \nonumber &= [\bbX^\top \circ \bbI_N \cdots (\bbS^{L-1}\bbX)^\top \circ \bbI_N ] \vec(\bbP)\\
    &= \bbA \vec(\bbP)
\end{align}
where $\bby=[\bby_1, \ldots, \bby_T]^\top$,   $\bbA=[\bbX^\top \circ \bbI_N \cdots (\bbS^{L-1}\bbX)^\top \circ \bbI_N ]$ and $\circ$ denotes the Khatri-Rao product. An estimate  of $\bbP$ can then be obtained as:
\vspace{-.2cm}
\begin{align}
\label{eq:pseudo-inverse}
    \bbP= \operatorname{unvec}(\bbA^\dagger \bby),
\end{align}
where $\bbA^\dagger= (\bbA^H \bbA)^{-1}\bbA^H$ is the pseudo-inverse of $\bbA$.


\subsection{Graph Filter   Estimation with Output-only Data}
\label{subsec:learn-P-output-only}

The goal of this section is to estimate the graph filter coefficients $\bbP$ with the only knowledge of the output graph signals $\bbY \in \reals^{N \times T}$. We assume that each $\bby_t$ can be modelled as the output of a NV-GF when excited with a white input $\bbx_t~\sim ~ \ccalN(\boldsymbol{0}, \bbI)$, which is not directly observable. A viable approach is then to fit the (empirical) second order information of the process which, together with the filter parametrization \eqref{eq:node-variant}, leads to the following optimization problem:
\begin{align}
\label{eq:main-problem}
    \min_{\bbP} \|\hat{\bbC}_y - \bbH_I(\bbP, \bbS)\bbH_I(\bbP, \bbS)^\top\|_F^2
\end{align}
where $\hat{\bbC}_y= (1/T)\bbY\bbY^\top$ is the sample covariance matrix ($\bbY$ is already centered).

Problem~\eqref{eq:main-problem} is non-convex in $\bbP$; thus, to alleviate the non-convexity, we consider instead the simpler (yet, again non-convex) problem:
\begin{align}
\label{eq:main-problem-approximated}
    &\min_{\bbP, \bbU} \|\bbR - \bbH_I(\bbP, \bbS)\bbU\|_F^2 \quad \text{s.t.} \; \bbU^\top\bbU=\bbI
\end{align}
where $\bbR$ is a square matrix such that $\hat{\bbC}_y= \bbR\bbR^\top$, and $\bbU$ is an $N \times N$ orthogonal matrix. By exploiting the SVD of the matrix $\bbY$, i.e.,  $\bbY=\bbU_y\bbLambda_y\bbV_y^\top $,  possible choices for $\bbR$ are $\bbR=(1/\sqrt{T})\bbU_y\bbLambda_y \bbU_y^\top$ and $\bbR=(1/\sqrt{T})\bbU_y\bbLambda_y$. 
Notice that if  \eqref{eq:main-problem} has a solution, then \eqref{eq:main-problem-approximated} has a solution.

\sn
Despite the fact that  \eqref{eq:main-problem-approximated} is not jointly convex  in $\bbP$ and $\bbU$, it is convex in $\bbP$ for a fixed $\bbU$; moreover, for a fixed $\bbP$, it reduces to the well-studied orthogonal Procrustes problem~\cite{green1952orthogonal}, for which a closed form solution exists albeit its non-convexity. Thus, an alternating minimization over $\bbP$ and $\bbU$ can be put forth, as follows: \textit{a)} first, given the estimate of $\bbU$ at the $(n-1)$th iteration, i.e., $\bbU^{(n-1)}$, the estimation problem at the $n$th iteration for the filter taps matrix $\bbP$ reads as:
    \begin{align}
\label{eq:step-A}
    &\bbP^{(n)}=\argmin_{\bbP} \|\bbR - \bbH_I(\bbP, \bbS)\bbU^{(n-1)}\|_F^2, 
    \tag{36a}
\end{align}
which can be solved in closed form by considering $\bbY = \bbR$ and $\bbX=\bbU^{(n-1)}$ in~\eqref{eq:vectorization-P-estimation}. The solution of \eqref{eq:step-A} is then  \textit{b)} used in the next step to refine the estimate of the unitary matrix $\bbU$, i.e.,:
        \begin{align}
\label{eq:step-B}
    \bbU^{(n)}= \argmin_{\bbU} & \|\bbR - \bbH_I(\bbP^{(n)}, \bbS)\bbU\|_F^2 \nonumber \\ 
    & \text{s.t.} \; \bbU^\top\bbU=\bbI
    \tag{36b}
\end{align}
for which the closed form solution is $\bbU^{(n)} = \bbV_p\bbU_p^\top$, where $\bbU_p$ and $\bbV_p$ are the left and right singular vector matrices, respectively,  of the matrix product $\bbR^\top \bbH_I(\bbP^{(n)}, \bbS)$, that is,  $\bbR^\top \bbH_I(\bbP^{(n)}, \bbS) = \bbU_p \bbLambda_p \bbV_p^\top$, with $\bbLambda_p$ the matrix of singular values. The solution of~\eqref{eq:step-B} is then fed again into~\eqref{eq:step-A}, unless a predefined number of iterations or stopping criterion is reached.

\subsection{Dual Graph Frequency Estimation}
\label{subsec:learn-lambda}

The next step is to learn the dual graph frequencies $\bblambda_f$, which are the only unknowns of the dual GSO $\bbS_f$. In order to do this, consider again the matrix form of~\eqref{eq:primal-coefficients}:
\begin{align}
\label{eq:matrix-factorization}
    \bbP=\bbPsi_f(\bblambda_f)\bbC 
\end{align}
where we explicitly wrote $\bbPsi_f$ as a function of $\bblambda_f$ to highlight the fact that the matrix is entirely determined by its second column $\bblambda_f$, representing our unknown. Then, the problem we aim to solve can be formally stated as follows:

\sn
\begin{tcolorbox}
    \textit{Given the matrix $\bbP \in \reals^{N \times L}$, recover the input vector $\bblambda_f \in \reals^N$ and the coefficient matrix $\bbC \in \reals^{K \times L}$ such that~\eqref{eq:matrix-factorization} holds as accurately as possible.}
\end{tcolorbox}

\sn
Although the problem can be approached from a pure algebraic point of view as a structured matrix factorization,  a pleasing geometrical interpretation of~\eqref{eq:matrix-factorization} is given by interpreting the vectors $\bbp_0, \ldots, \bbp_{L-1}$  as function values obtained by sampling $L$ distinct polynomials $p_0(\lambda_f), \ldots, p_{L-1}(\lambda_f)$, all with  degree $K-1$, in the same $N$ unknown locations $\lambda_{f,0}, \ldots, \lambda_{f, N-1}$. The goal is to recover the original locations (and polynomial coefficients) from the available sampled function values. The only side information we have about these sampling points is \textit{i)} which function they belong to and \textit{ii)} that they are ordered in such a way that the related sampling points are aligned. See~\cite{natali2023blind} and the Supplemental Material for an illustration.

\sn
\textbf{Subspace Fitting.}
Assume we start by considering the following optimization problem to estimate $\bblambda_f$ and $\bbC$:
\begin{align}
\label{eq:original}
    \min_{\bblambda_f, \bbC} \;\frac{1}{2} \|\bbP - \bbPsi_f(\bblambda_f)\bbC\|_F^2,
\end{align}
which can be solved, for instance, with an alternating minimization approach, without guarantee of convergence to a global optimum.

\sn
An alternative subspace method can be devised when $L \geq K$. We then consider the economy-size SVD of matrix $\bbP$, i.e.,  $\bbP=\bbU\bbSigma\bbZ^\top$, where $\bbU \in \reals^{N \times K}$ and $\bbZ \in \reals^{L \times K}$  are the left and right singular vectors, respectively, and $\bbSigma \in \reals^{K \times K}$ is the diagonal matrix of singular values. Since both $\bbPsi_f(\bblambda_f)$ and  $\bbU$ represent a basis for the column space of $\bbP$, there exists a non-singular matrix $\bbQ \in \reals^{K \times K}$ such that 
$\bbPsi_f=\bbU\bbQ$. The subspace fitting~\cite{viberg1991sensor} problem reads then as:
\begin{align}
\label{eq:subspace-fitting}
    \min_{\bblambda_f, \bbQ}\;\frac{1}{2} \|\bbPsi_f(\bblambda_f) - \bbU\bbQ\|_F^2,
\end{align}
which, upon substituting the pseudoinverse solution  $\bbQ=\bbU^\dagger\bbPsi_f$ into~\eqref{eq:subspace-fitting}, can be casted as the following equivalent problem:
\begin{align}
\label{eq:equivalent-subspace-fitting}
    \min_{\bblambda_f}\; \{ f(\bblambda_f):=\frac{1}{2} \|\boldsymbol{\Pi}\bbPsi_f(\bblambda_f)\|_F^2\},
\end{align}
with $\boldsymbol{\Pi}:=\bbI_N - \bbU\bbU^\dagger$ the orthogonal projection matrix onto the orthogonal complement  of $\bbU$.
In other words, problem~\eqref{eq:equivalent-subspace-fitting} aims to find a vector $\bblambda_f$ such that the Vandermonde matrix $\bbPsi_f(\bblambda_f)$ is orthogonal to the subspace spanned by the orthogonal complement of $\bbU$. Notice that we require $L \geq K$, so that the matrix $\bbP$ can have rank $K$ revealing the subspace of the Vandermonde matrix $\bbPsi_f(\bblambda_f)$ we want to estimate.

%

\sn
Problem~\eqref{eq:subspace-fitting} is not convex in $\bblambda_f$ due to the polynomial degree $K$ (unless $K=2$, i.e., the model in~\eqref{eq:matrix-factorization} is linear in $\bblambda_f$). To tackle the non-convexity of the problem, we resort to sequential convex programming (SCP)~\cite{boyd2008sequential}, a local optimization method that leverages convex optimization, where the non-convex portion of the problem is modeled by convex functions that are (at least locally) accurate. As in any non-convex problem, the initial starting point plays a big role; thus, if no prior information on the variable is given, a multi-starting point approach is advisable. The SCP formulation can be found in~\cite{natali2023blind} and in the Supplemental Material (and potentially in the final version of this manuscript).

\begin{remark}
\label{remark:music}
The nature of the problem and the formulation~\eqref{eq:equivalent-subspace-fitting} share similarities with the MUSIC algorithm~\cite{schmidt1986multiple}. There, however, the problem considers $\bbPsi_f^\top$ instead of $\bbPsi_f$ and $N$ independent $1$-dimensional searches can be carried out to find $\bblambda_f$ (which would be contained in the second row of $\bbPsi^\top$). In our case, each column contains all the $N$ variables and an $N$-dimensional search is needed, rendering a ``scanning'' of the vector variable $\bblambda_f$ infeasible, unless $N$ is very small.
\end{remark}

\subsection{Ambiguities}
\label{sec:ambiguities}
Notice that \eqref{eq:matrix-factorization} is not free of model ambiguities, since different pairs $(\bblambda_f,\bbC)$ may lead to the same observation matrix $\bbP$. Because we assume that $\bbP$ has rank $K$, if $\bbPsi_f$ and $\bbC$ are the true matrix factors satisfying~\eqref{eq:matrix-factorization}, then for any $K \times K$ invertible matrix $\bbT$, it holds:
\begin{align}
\label{eq:model-ambiguity}
    \bbP= \bbPsi_f\bbT\bbT^{-1}\bbC= \bbPsi_f^\prime \bbC^\prime.
\end{align}
However, $\bbPsi_f^\prime$ needs to be Vandermonde in order for~\eqref{eq:model-ambiguity} to respect the model structure in~\eqref{eq:matrix-factorization}. As we proved in~\cite[Theorem 2]{natali2023blind}, matrix $\bbT$ needs in general to be a Pascal matrix.
The consequence of such theorem is that without any added constraint on $\bbPsi_f$ or $\bbC$ in~\eqref{eq:matrix-factorization},  every shifted and scaled version of the groundtruth parameter $\bblambda_f$ (for an appropriate $\bbC$), perfectly fits the observation model and is a valid solution for~\eqref{eq:equivalent-subspace-fitting}. This is satisfactory for our purpose: a shift and scale of the graph eigenvalues maintains the same topological
structure of the original graph (removing the self loops caused by the shift); we will illustrate this in  Section~\ref{sec:sims}.

\section{Numerical Experiments}
\label{sec:sims}

In this section we perform numerical experiments relying on the theory outlined in Section~\ref{sec:learning}. Specifically, we first test our learning algorithm on synthetic data to assess the validity of the approach and its drawbacks; finally, we use it on real data.

\subsection{Synthetic Data}
\label{sec: synthetic}
The goal of this section is to validate the approach outlined in Section~\ref{sec:learning}, by assuming the knowledge of the dual eigenvalues $\bblambda_f$, which however is not used during the training phase but only used to compute the performance metrics. We have the following \textit{generation} and \textit{learning} phases (also depicted in the Supplemental Material for visual clarity):

\sn
\textbf{1) Generation.}
 \begin{enumerate}
     \item[(a)] \textbf{Dual eigenvalues} $\bblambda_f$: in order to have control over the spacing between different eigenvalues, we generate them with the following strategy. First, we uniformly sample the eigenvalue domain leading to the grid $\bbu:=[u_0, u_1, \ldots, u_{N-1}]^\top$, where the distance among two samples $P:= u_n - u_{n-1}$ is constant for all $n$. Then, the actual eigenvalues $\bblambda_f$ are generated as $\lambda_{f, n} = u_n + j_n$ where $ j_n \sim \ccalT\ccalN(0, (\delta P/2)^2)$, with $\ccalT\ccalN(0,\sigma_j^2)$ a Gaussian distribution with zero mean and standard deviation $\sigma_j$, yet truncated at $\sigma_j$. The (jitter) parameter $\delta>0$ specifies the randomness of the eigenvalues. If $\delta < 1$ the eigenvalues maintain the same ranking order as the uniform samples, while if $\delta \geq 1$ they can overlap to the adjacent ones. Notice that in this synthetic setup we are not interested in creating a ``meaningful'' dual graph, but rather assessing whether our algorithmic routine is able to identify such dual graph.

     \item[(b)] \textbf{Filter taps} $\bbP$: we generate the primal node-varying filter taps $\bbP \in \reals^{N \times L}$ as $\bbP= \bbPsi_f \bbC$, with $\bbPsi_f$  the Vandermonde matrix associated to $\bblambda_f$ generated in step (a), and $\bbC \in \complex^{K \times L}$  a random expansion coefficient matrix [cf.~\eqref{eq:duality-coefficients}] generated as ${\rm vec}({\bf C}) \sim \ccalN(\boldsymbol{0}, \bbI_{KL})$. To increase the curvature of the considered polynomials (thus avoiding to have almost flat curves in the domain of interest), we perform an  extra weighting scheme by increasing the weight of higher order monomials; that is, we multiply the matrix $\bbC$ with a mask matrix as $\bbC \leftarrow [\boldsymbol{1}, 2 \boldsymbol{1}, \ldots, K \boldsymbol{1}]^\top \odot \bbC$, where the $i$th column of the mask matrix is the constant vector containing in all its entries the value $i$. 

     \item[(c)] \textbf{Input-Output Data} $\bbX, \bbY$: we generate $T$ input graph signals $\bbx_t \sim \ccalN(\boldsymbol{0}, \bbI_N)$  and stack them in the matrix $\bbX=[\bbx_1, \ldots, \bbx_T]$. Then, we filter $\bbX$ with the type-I NV-GF $\bbH_I(\bbP, \bbS)$ to obtain $T$ new  (possibly noisy) graph signals $\bbY=[\bby_1, \ldots, \bby_T]= \bbH_I(\bbP, \bbS)\bbX + [ {\bf n}_1, \dots, {\bf n}_T ]$, with $\bbn_t \sim \ccalN(\boldsymbol{0}, \sigma^2\bbI_N)$ the measurement noise.
 \end{enumerate}

 \sn
 \textbf{2) Learning.} The goal in this phase is to recover the original $\bblambda_f$ of step 1(a) from the input-output data $\{\bbX, \bbY\}$. This is achieved with the algorithmic routine introduced in Section~\ref{subsec:learn-P} and Section~\ref{subsec:learn-lambda}; namely:
 \begin{enumerate}
     \item[(a)] \textbf{Filter Taps $\Tilde{\bbP}$}: we find an estimate $\Tilde{\bbP}$ of the filter taps $\bbP$ through~\eqref{eq:pseudo-inverse};

     \item[(b)] \textbf{Dual eigenvalues} $\Tilde{\bblambda}_f$: 
 we find an estimate $\Tilde{\bblambda}_f$ of $\bblambda_f$ with the procedure outlined in Section~\ref{subsec:learn-lambda}, namely, casting the problem as a subspace fitting problem and solving it with SCP. During this step, we also estimate the coefficient matrix $\bbC$ of the expansion model leading to $\tilde{\bf C}$.
 \end{enumerate}

 \begin{figure*}
    \begin{subfigure}{0.33\textwidth}
        \includegraphics[scale=0.43]{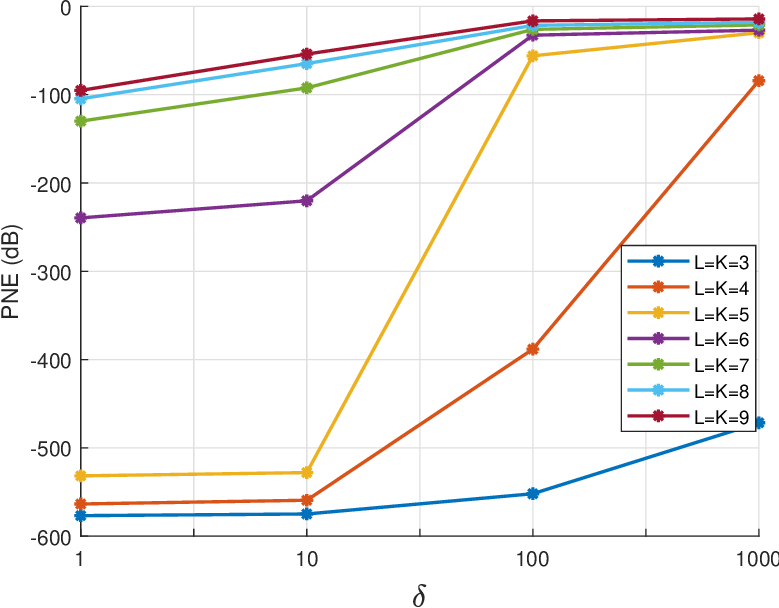}
    \end{subfigure}%
    \begin{subfigure}{0.33\textwidth}
        \includegraphics[scale=0.43]{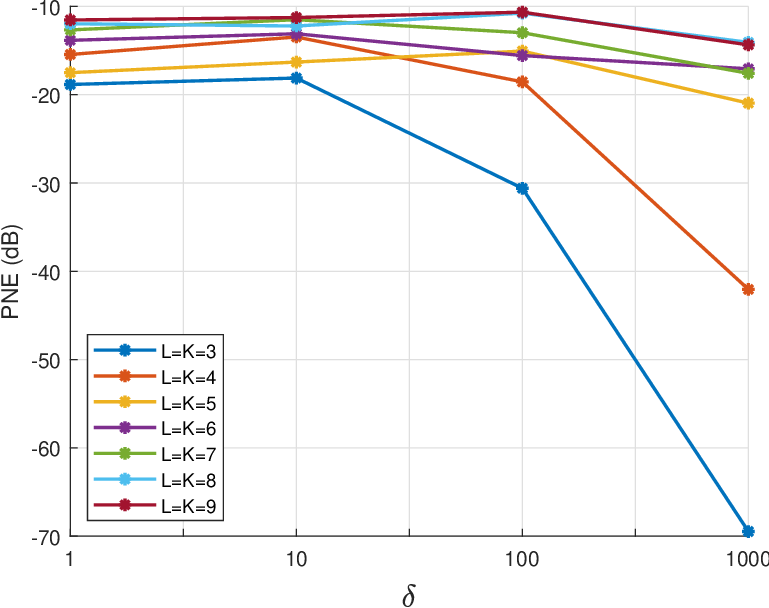}
    \end{subfigure}%
    \begin{subfigure}{0.33\textwidth}
        \includegraphics[scale=0.43]{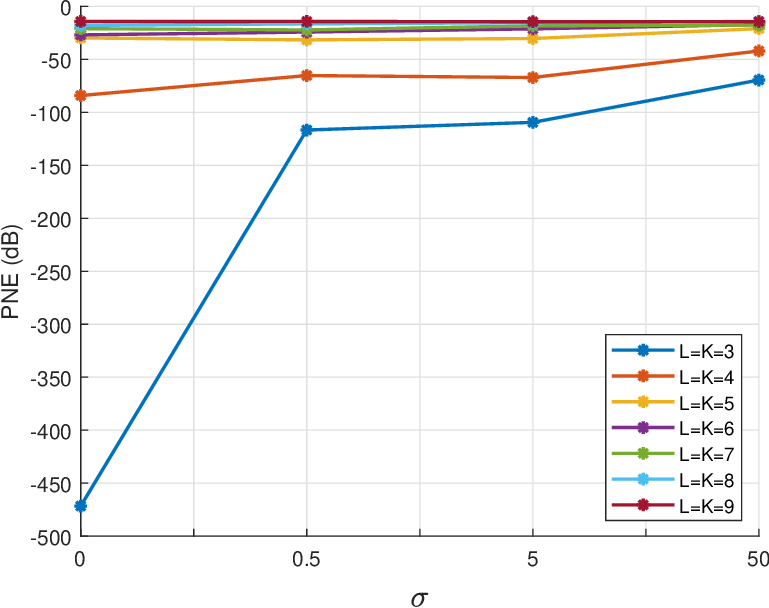}
    \end{subfigure}%
    \caption{PNE (in dB) as a function of $\delta$ with $\sigma=0$ (left) and $\sigma=50$ (center), for different values of $L,K$; and  PNE (in dB) as a function of $\sigma$ for $\delta=1000$ (right).}
    \label{fig: PNE_as_function}
\end{figure*}

\sn 
\textbf{Metrics.}
To assess the validity of the proposed approach, we consider two different performance metrics, one relative to the estimation of the filter coefficients $\bbP$ and one relative to the estimation of the dual eigenvalues $\bblambda_f$. For the filter coefficients, we consider the normalized squared error (NSE), computed as:
\begin{align}
\label{eq:NSE}
    \text{NSE}(\Tilde{\bbP}, \bbP) = \frac{\|\Tilde{\bbP}-  \bbP\|_F^2}{\|\bbP\|_F^2}.
\end{align}
For the dual eigenvalues, recall that we can recover the solution of problem~\eqref{eq:equivalent-subspace-fitting} up to a shift and scaling of the true positions [cf. Section~\ref{sec:ambiguities}]. Thus, as performance metric, we use  the normalized error modulo Pascal (PNE), defined as:
\begin{align}
\label{eq:pascal-error}
    \text{PNE}(\Tilde{\bblambda}_f, \bblambda_f)= \min_{t_0, t_1} \frac{\| \bblambda_f - (t_0 \boldsymbol{1} + t_1 \Tilde{\bblambda}_f)\|_2^2}{\|\bblambda_f\|_2^2}
\end{align}
  which measures how far the true eigenvalues are from a linear transformation of the recovered estimates.
 Clearly~\eqref{eq:pascal-error} is zero whenever $\Tilde{\bblambda}_f$ is a solution for~\eqref{eq:equivalent-subspace-fitting}.

 \sn
\textbf{Results.} We generate our primal graph $\bbS$ as a random sensor network with $N=40$ nodes, and set $u_0= -1$ and $u_{N-1}= +1$ [cf.~1(a)]. We run the algorithm for different parameter configurations, specifically: the order of the filter in the primal domain $L \in \{2, \ldots, 9\}$, which is also the polynomial degree of the primal eigenvalues (cf.~\eqref{eq:dual-node-variant}); the order of the filter in the dual domain $K \in \{2, \ldots, 9\}$, with $K \leq L$, which is also the polynomial degree of the dual eigenvalues; the jitter parameter $\delta \in \{1, 10, 100, 1000\}$; and the noise standard deviation $\sigma \in \{0, 0.5, 5, 50\}$. The number of input (and output) graph signals is set to $T=3000$. Due to the non-convexity of the cost function~\eqref{eq:equivalent-subspace-fitting}, we run the algorithm with $5$ different starting points $\bblambda_f^0$, one of which is the uniform grid $\bbu$; this, together with the jitter parameter $\delta$ [cf. Generation (a)] helps us also understanding the ``magnitude'' of the objective function's non-convex landscape: if the objective function is highly non-convex, even an initial starting point $\bblambda_f^0$ close to the real $\bblambda_f$ (meaning a small $\delta$) might incur a very high objective value and likely end-up in a local minimum. In such a case, a random starting point might be beneficial. The magnitude of non convexity of the objective function increases by increasing $K$. Finally, we compute the performance metrics relative to the solution associated to the different starting points.

\sn
In Fig.~\ref{fig: PNE_as_function} we show the PNE (in dB) as a function of $\delta$ for different $L=K$\footnote{For all the plots we set $L=K$ for visualization clarity. To reproduce the experiments with different parameter settings, we make available the source code:  \url{https://github.com/albertonat/genConv}}, in the noiseless case (left figure) and noisy case (middle figure); in addition, we show (right) how the PNE varies as a function of the noise $\sigma$ for fixed $\delta = 1000$. We can make the following observations: as expected, for low-degree polynomials, the algorithm performs better, since the non-convexity of the problem increases with increasing polynomial order. This is visible from the left and middle figure: for a small perturbation $\delta$, also high orders yield a good performance; which degrades by increasing $\delta$. However, when noise is present in the observations, a  random initial starting points (i.e. higher $\delta$) seems to be beneficial. Moreover, noise in the measurement (right figure) has obviously a negative impact in the learning performance, which however enables us to reconstruct the graph as we can see next.


 \begin{figure*}[]
\begin{subfigure}{0.33\textwidth}
    \centering
\includegraphics[scale=0.45, trim= 0cm 0cm 0cm 0cm, clip=true ]{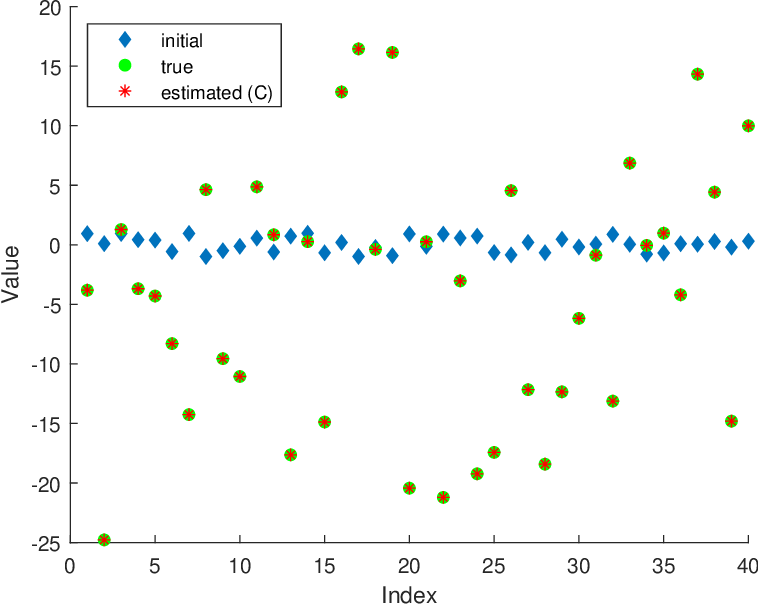}
\end{subfigure}%
\begin{subfigure}{0.33\textwidth}
    \centering
\includegraphics[scale=0.45,  trim= 0cm -.5cm 0cm 0cm, clip=false]{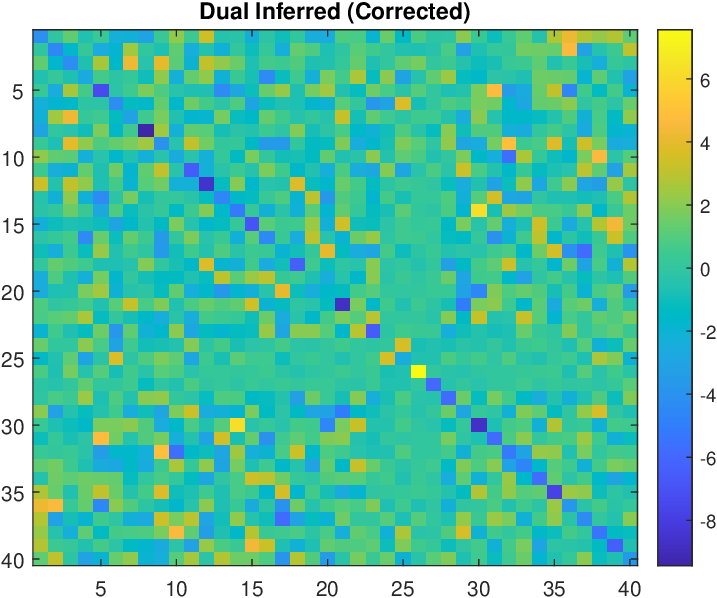}
\end{subfigure}%
\begin{subfigure}{0.33\textwidth}
\centering
\includegraphics[scale=0.45, trim= 0 -.5cm 0cm 0cm, clip=false ]{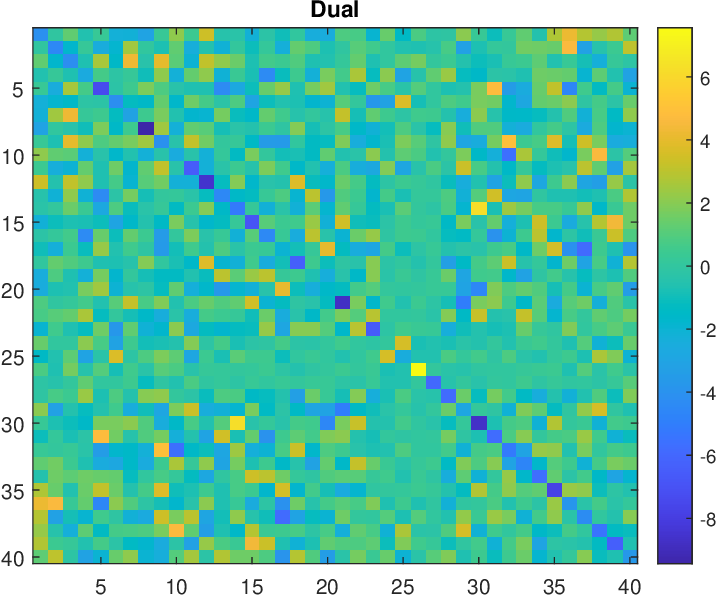}
\end{subfigure}%
\\
\begin{subfigure}{0.33\textwidth}
    \centering
    \includegraphics[scale=0.45, trim= 0cm 0cm 0cm 0cm, clip=true ]{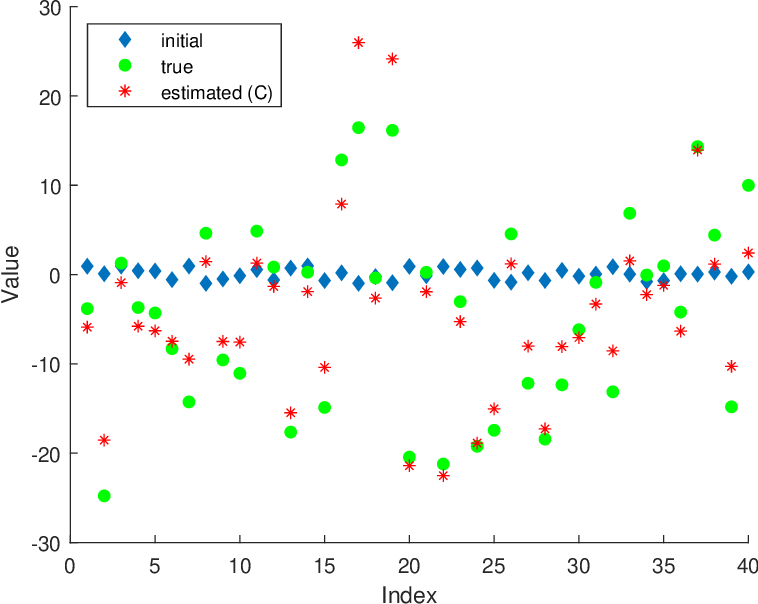}
\end{subfigure}%
\begin{subfigure}{0.33\textwidth}
    \centering
    \includegraphics[scale=0.45, trim= 0 -.5cm 0cm 0cm, clip=false]{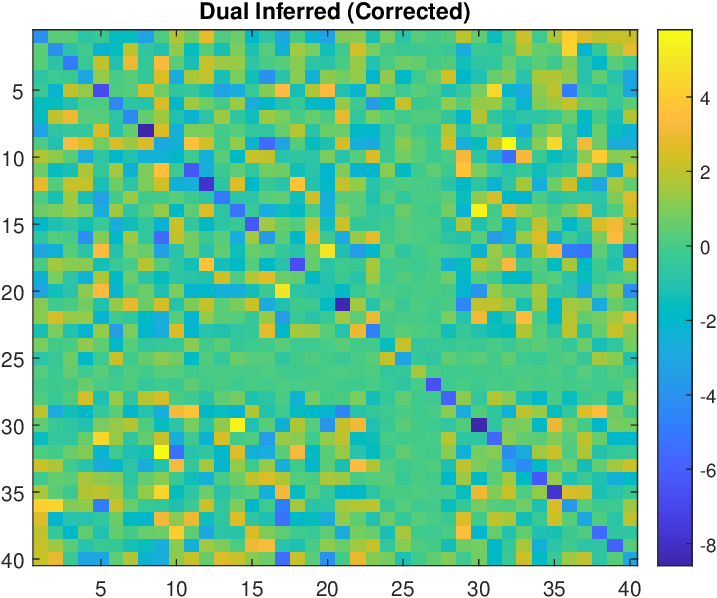}
\end{subfigure}%
\begin{subfigure}{0.33\textwidth}
\centering
\includegraphics[scale=0.45, trim= 0 -.5cm 0cm 0cm, clip=false ]{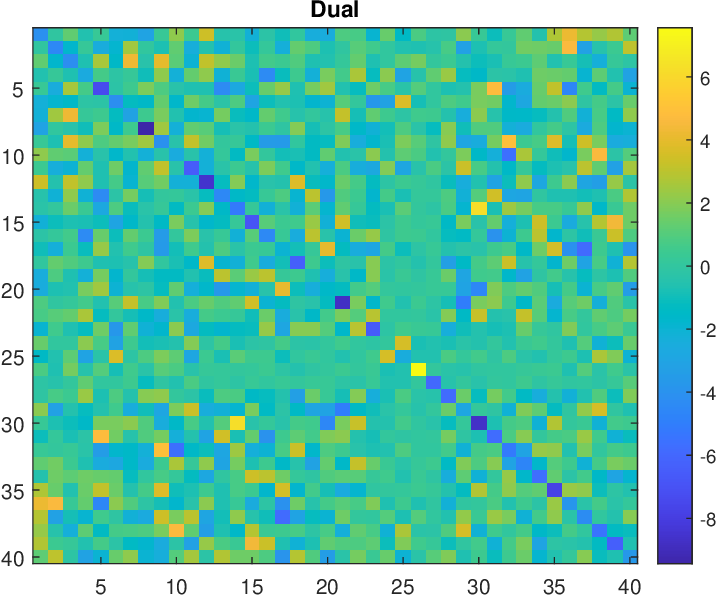}
\end{subfigure}%
\vspace{-0.1cm}
\caption{(Noiseless case $\sigma=0$) Results for $L=K=3$ (top row) and $L=K=9$ (bottom row) with $\delta = 1 \times 10^4$. (Left Column) True eigenvalues (green circles),
inferred eigenvalues (red crosses) and initial starting point of the algorithm (blue diamonds); (Center Column) Inferred dual GSO with ambiguity-correction; (Right Column) True dual GSO $\bbS_f$.}
\label{fig:noiseless}
\vspace{-.1cm}
\end{figure*}

\sn
To visually assess the algorithm's performance, in Fig.~\ref{fig:noiseless} we show, for the noiseless case, the ambiguity-corrected estimated eigenvalues $\Tilde{\bblambda}_f^c$ (red crosses), together with the original eigenvalues $\bblambda_f$ (green circles) and the initial starting point of the algorithm $\bblambda_f^0$ (blue diamonds) for $L \! =\!K\!=\!3$ (top row) and $L \! = \! K\!=\!9$ (bottom row), both with $\delta \! = 1 \!\times\!10^4$, which corresponds to a completely random configuration of the eigenvalues. The ambiguity correction is explicitly performed as:
\vspace{-.2cm}
\begin{align}
    \Tilde{\bblambda}_f^c= [\boldsymbol{1}  \; \Tilde{\bblambda}_f] [\boldsymbol{1}  \; \Tilde{\bblambda}_f]^\dagger  \bblambda_f,
\end{align}
 which is the closest point to $\bblambda_f$ up to a linear transformation dictated by the optimal $t_0$ and $t_1$ minimizing~\eqref{eq:pascal-error}.

\sn
In the $L=K=3$ case, the NSE is $\text{NSE}(\Tilde{\bbP}, \bbP)=2.41 \times 10^{-29} $  and the PNE is $\text{PNE}(\Tilde{\bblambda}_f, \bblambda_f)= 3.34 \times 10^{-25}$; in the $L=K=9$ case the NSE is $\text{NSE}(\Tilde{\bbP}, \bbP)= 9.20 \times 10^{-23}$ and the PNE is $\text{PNE}(\Tilde{\bblambda}_f, \bblambda_f)= 1.03 \times 10^{-1}$. In both cases the inferred dual GSO, shown in the middle column of Fig.~\ref{fig:noiseless}, correctly reveals the structure of the true dual GSO $\bbS_f$, shown in the right column; this even though the eigenvalue reconstruction (left column) is perfect only in the first case. In other words, even with an inexact (but not random) reconstruction of the eigenvalues, the algorithm seems to adequately capture the connections present in the true GSO $\bbS_f$. The reason behind the difference in error estimation among the two cases is mainly due to the high polynomial degree $K$, which on the analytic side renders the objective function~\eqref{eq:equivalent-subspace-fitting} highly non-convex (and hence easier for the algorithm to end up in a local minimum); on the algebraic side it increases the numerical instability of performing the pseudoinverse of the matrix $\bbA$ required for a correct estimation of the filter parameter matrix $\bbP$. Hence, even a perfect dual graph frequency estimation step fitting perfectly the estimated $\bbP$, might fail to perfectly reconstruct the true eigenvalues $\bblambda_f$. For $\delta=10$, the algorithm improves the PNE by two orders of magnitude and the inferred eigenvalues $\Tilde{\bblambda}_f$ nearly overlap the true ones $\bblambda_f$.

\sn
For the noisy scenario, in Fig.~\ref{fig:noise-case} we show the results obtained by considering the two cases described above (i.e. $L=K=3$ and $L=K=9$ with $\delta=1000$), but with a measurement noise having $\sigma=50$. In this case the NSE is $8.88 \times 10^{-6}$ and the PNE is $8.72 \times 10^{-5}$, while for the latter the NSE is $9.74 \times 10^{-16}$ and the PNE is $2.00 \times 10^{-1}$. The eigenvalues and graph reconstruction is successful despite the fact that the reconstruction of the filter taps is not perfect as in the noiseless case, which gives us hope for the robustness of the algorithm when measuring noisy data.

\sn
Overall we can state that from an algorithmic point of view, the algorithm is robust in presence of a considerable eigenvalue perturbation $\delta$ and noise $\sigma$, especially for low polynomial degree. In instances where the polynomial degree is high, the inherent non-convex nature of the problem introduces substantial complexity into the optimization process, which however leads to a dual graph resembling the original one.  the  In addition, we observe that when the NSE is high, meaning that  $\bbP$  has not been properly reconstructed, the  PNE is also usually high, which is somehow expected since we rely on $\bbP$ to estimate the dual eigenvalues $\bblambda_f$.

\begin{figure*}
\begin{subfigure}{0.33\textwidth}
    \centering
\includegraphics[scale=0.45]{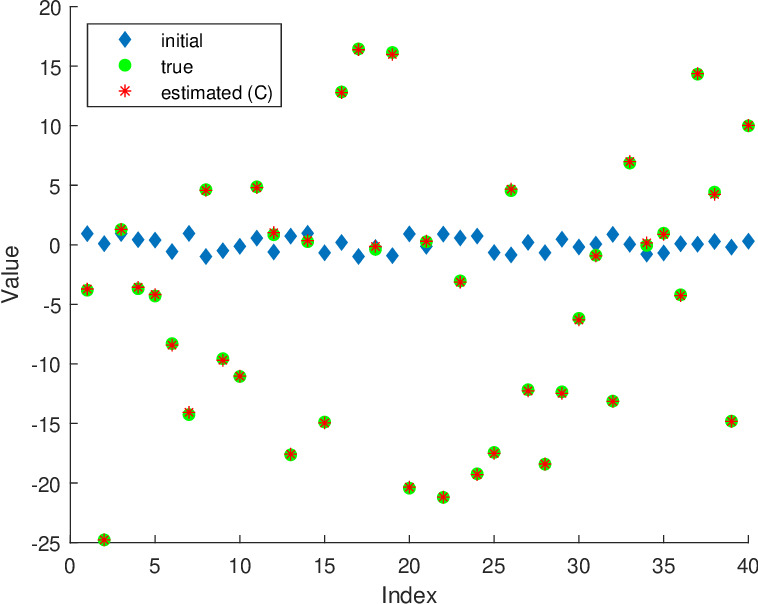}
\end{subfigure}%
\begin{subfigure}{0.33\textwidth}
    \centering
\includegraphics[scale=0.45, trim= 0 -.5cm 0cm 0cm, clip=false]{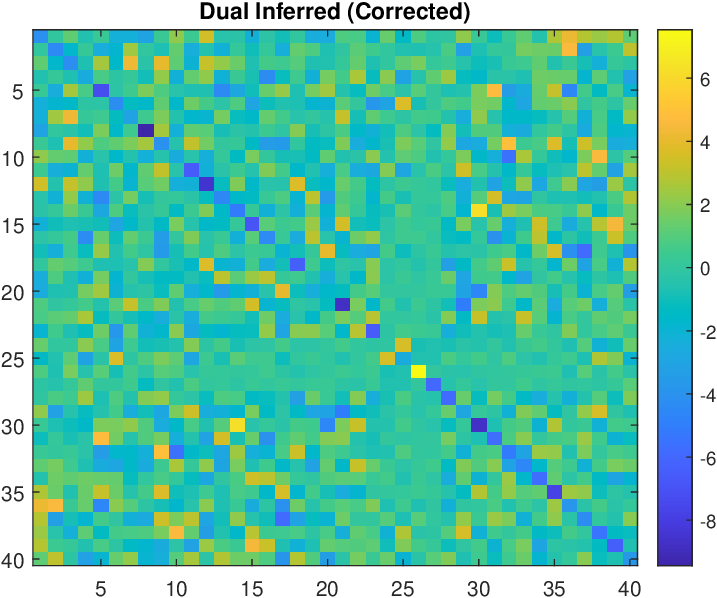} 
\end{subfigure}%
\begin{subfigure}{0.33\textwidth}
\centering
\includegraphics[scale=0.45, trim= 0 -.5cm 0cm 0cm, clip=false]{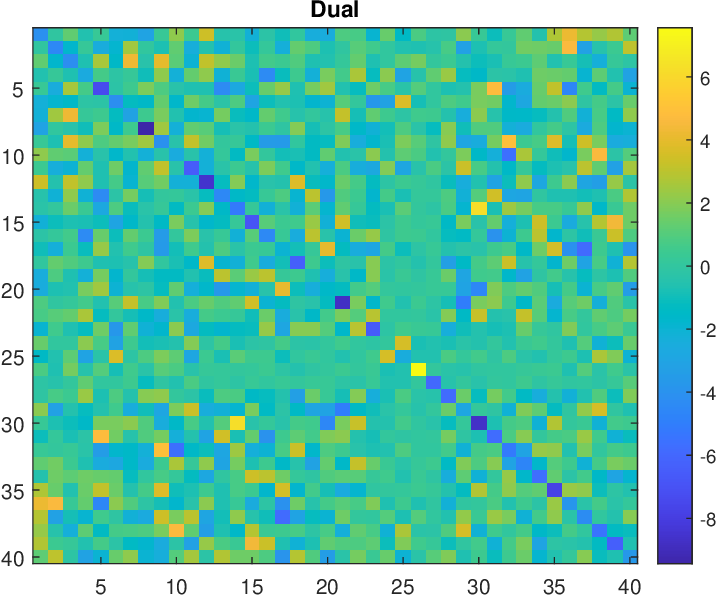}
\end{subfigure}%
\\
    \begin{subfigure}{0.33\textwidth}
        \includegraphics[scale=0.45]{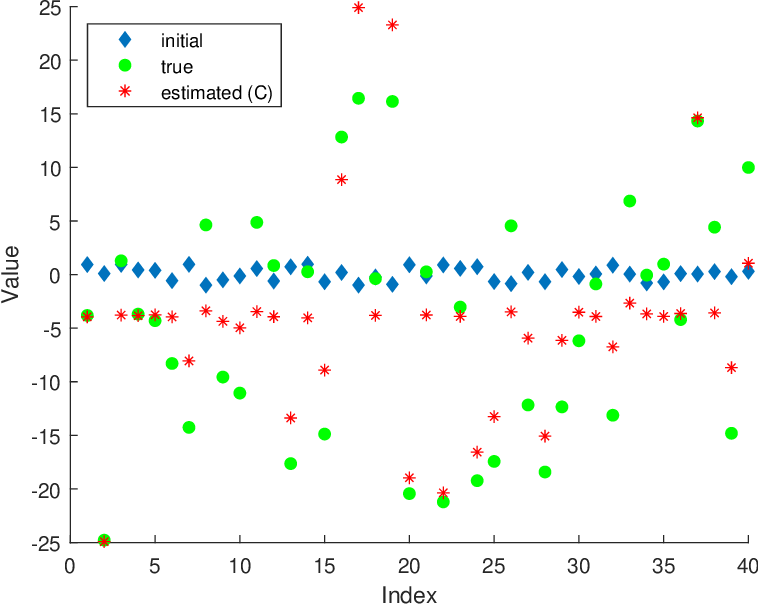}
    \end{subfigure}
       \begin{subfigure}{0.33\textwidth}
        \includegraphics[scale=0.45,trim= 0 -.5cm 0cm 0cm, clip=false]{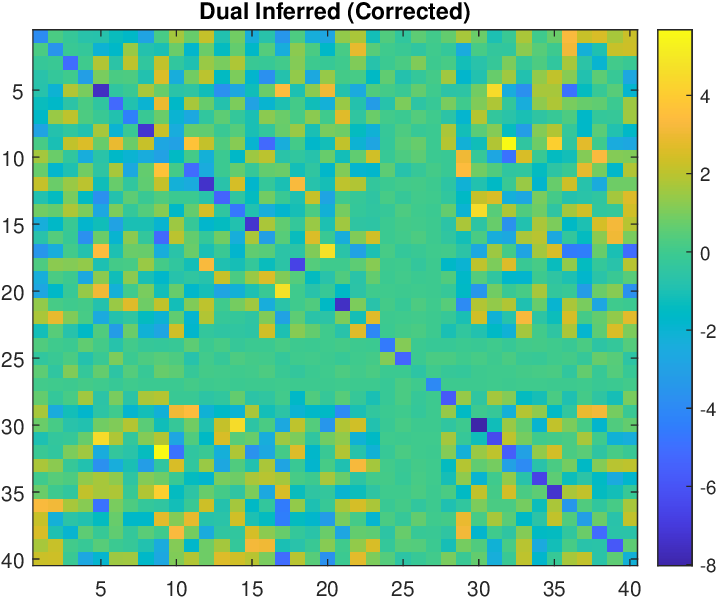}
    \end{subfigure}
       \begin{subfigure}{0.33\textwidth}
        \includegraphics[scale=0.45, trim= 0 -.5cm 0cm 0cm, clip=false]{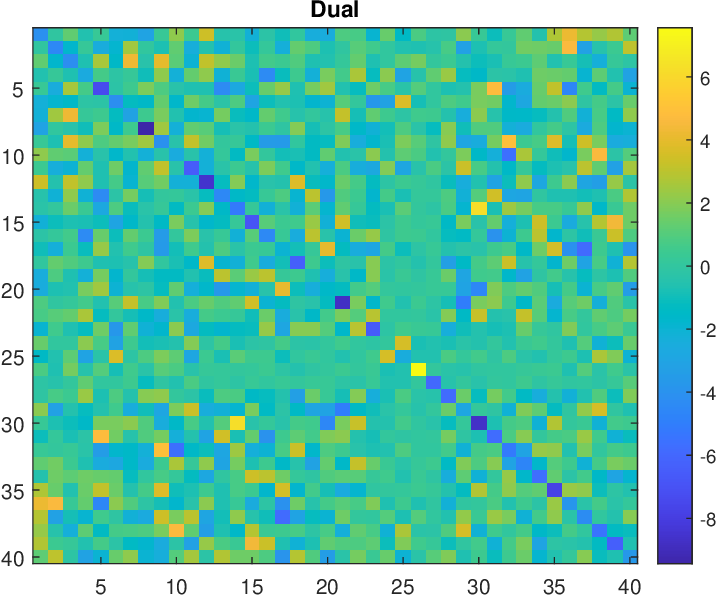}
    \end{subfigure}
    \vspace{-.5cm}
    \caption{ (Noisy case $\sigma=50$) Results for $L=K=3$ (top row) and $L=K=9$ (bottom row) with $\delta = 1 \times 10^3$. (Left Column) True eigenvalues (green circles), inferred eigenvalues (red crosses) and initial starting point of the algorithm (blue
diamonds;  (Center Column) Inferred dual GSO with ambiguity-correction; (Right Column) True dual GSO $\bbS_f$.}
\label{fig:noise-case}
\end{figure*}

%
%

\vspace{-.2cm}
\subsection{Real Data}
In this section we exploit the theory developed in Section~\ref{sec:theorem} and Section~\ref{sec:learning} to infer a dual graph from real data. To evaluate the stationarity level of the given data $\bbY$ for a given GSO $\bbS$, we use the proxy-measure $\rho = \|\diag(\bbV^{-1} \bbC_y \bbV)\|_2^2 / \|\bbV^{-1} \bbC_y \bbV\|_F^2$, which measures the  ``diagonal dominance" of the spectral covariance matrix; a value of $1$ indicates that the data are stationary.

\sn
As performance metrics we monitor two errors: \textit{i)} the NSE($\bbY, \bbH_I(\Tilde{\bbP}, \bbS)\bbX$) for the input-output case [cf.~\ref{subsec:learn-P}] or the NSE($\bbR, \bbH_I(\tilde{\bbP}, \bbS)\tilde{\bbU}$) for the output-only case [cf.~\ref{subsec:learn-P-output-only}] (recall the definition of the NSE in \eqref{eq:NSE});
 and  \textit{ii)} the ``corollary error'' [cf. Corollary~\ref{corollary: categoric} in ~\eqref{eq: categoric}]:
\begin{align}
\label{eq:corollary-error}
    \varepsilon_c = \frac{\|\bbV^{-1} \bbH_I(\tilde{\bbP}, \bbS) - \bbH_{II}(\tilde{\hat{\bbP}}, \bbS_f)\bbV^{-1}\|_F^2}{\|\bbV^{-1} \bbH_I(\tilde{\bbP}, \bbS)\|_F^2}
\end{align}
which assesses whether the inferred $\bbS_f$ is a ``valid'' dual graph consistent with our theorem,
i.e., how much  the upper and lower branches of Fig.~\ref{fig:duality} diverge from each other. To make things more clear, the $\tilde{\bbP}$ refers to the estimation of the primal filter tap matrix $\bbP$ either following~\eqref{eq:pseudo-inverse} or~\eqref{eq:step-A}, while $\tilde{\hat{\bbP}}$ refers to the estimation of the dual filter tap matrix $\hat{\bbP}$, which is here computed as $\tilde{\hat{\bbP}} = \bbPsi^\dagger \tilde{\bbC}^\top$, with $\tilde{\bbC} = \bbPsi_f(\tilde{\bblambda}_f)^\dagger \tilde{\bbP}$, where $\tilde{\bblambda}_f$ is the estimation of the dual eigenvalues $\bblambda_f$ solving~\eqref{eq:equivalent-subspace-fitting}.
All in all,  error \textit{i)} concerns the graph filter estimation, while \textit{ii)} concerns the dual graph estimation.

\subsubsection{Traffic Volume} we consider a subsampled version of the open  dataset in~\cite{zhao2019spatial} which contains $T=1259$  traffic volume measurements at intervals of $15$ minutes at $N=13$ sensor locations along two major highways in  Northern Virginia/Washington, D.C.; in addition, the physical (road) network is available, see Fig.~\ref{fig:road}(left).  We denote with $\bbS$ the adjacency matrix representing the given road network, and with $\bbY \in \reals^{N \times T}$  the (centered) graph signals corresponding to the traffic volume measurements. These signals exhibit a non-stationary behavior captured by $\rho = 0.54$.

\begin{figure}[h]
    \centering
    \begin{subfigure}{0.25\textwidth}%
    \includegraphics[scale=0.4, trim = 2.6cm 2cm 1.6cm 0.9cm, clip= true]{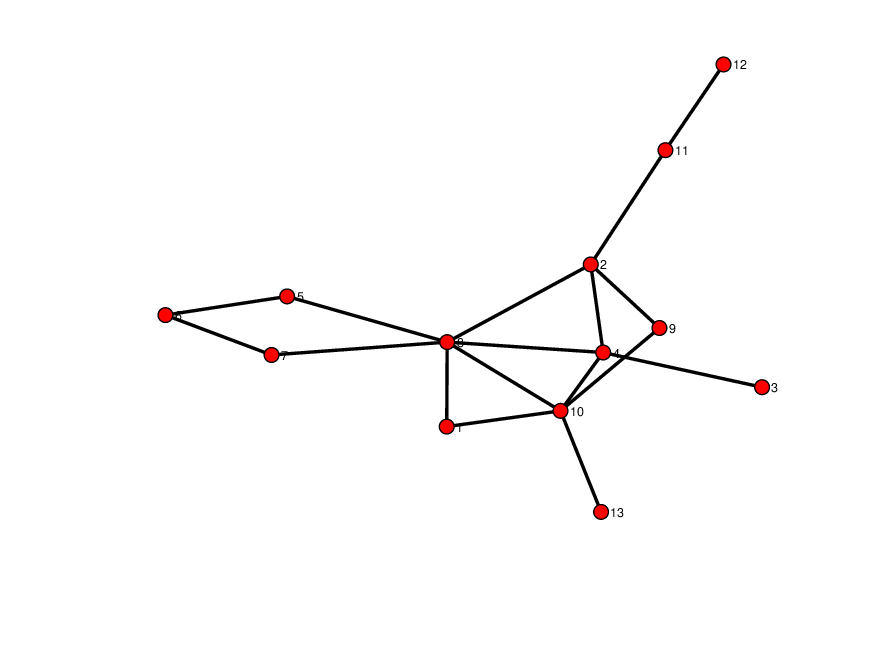}%
    \end{subfigure}%
    \begin{subfigure}{0.25\textwidth}%
    \vspace{.2cm}
    \includegraphics[scale=0.4]{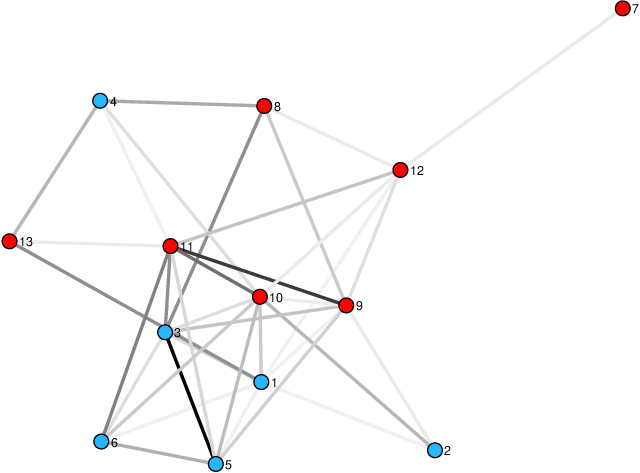}%
    \end{subfigure}%
        \vspace{-.2cm}
    \caption{(Left): primal graph of the road network; (Right) dual graph of the road network, where each node represents a graph frequency. Red (blue) nodes are the standard low  (high) frequencies.}
    \label{fig:road}
\end{figure}

In our experiment, we explore both input-output and output-only scenarios. 
In the former scenario, we define the input matrix $\bbX$ by aligning it with $\bbY$, shifting each column two positions to the left; consequently, the learning task revolves around forecasting the traffic volume $30$ minutes ahead. We initialize the parameters with $L=K=3$ and execute the algorithm. Notice that our data was deliberately not partitioned into training, validation, and test sets, since our primary focus is optimizing the fitting of our graph filter to the provided data, rather than evaluating its forecasting performance.

\vspace{-.3cm}
\sn
First, we learn the filter coefficients $\bbP$ for both scenarios, yielding a NSE equal to $8 \times 10^{-2}$ for both. Subsequently, we use the inferred $\tilde{\bbP}$ to learn the dual eigenvalues $\bblambda_f$ by solving~\eqref{eq:equivalent-subspace-fitting}. The associated dual GSO $\bbS_f = \bbV^{-1} \Diag(\tilde{\bblambda}_f) \bbV$, achieving a corollary error [cf.~\eqref{eq:corollary-error}] $\varepsilon_c = 1 \times 10^{-1}$ for scenario 1 and $\varepsilon_c = 3 \times 10^{-1}$ for scenario 2, is shown in Fig.~\ref{fig:road}(right), where we only display the $50\%$ biggest edges (in absolute values) to ease the visualization. We color with blue the first half of the nodes, representing what are usually considered ``high pass'' frequencies, and with red the second half of the nodes, representing what are usually considered the  ``low pass'' frequencies (remember that $\bbS$ is the adjacency matrix). Adjacent nodes in the graph are not necessarily among consecutive graph frequencies, as it is commonly assumed in  GSP. This shows that the frequency ordering that is commonly assumed does not fit our theory and we might obtain a more expressive way to embed the different graph frequencies.

\subsubsection{MNIST Dataset}
We consider the MNIST dataset of gray-scale handwritten digits from $0$ to $9$, focusing on the digit $5$\footnote{Similar results can be obtained with the other digits, see~\url{https://github.com/albertonat/genConv}} containing $T = 5949$ images of size $18 \times 18$. We model each pixel of the image as a node of the primal graph $\bbS$ and its pixel intensity as the graph signal value at that node. As a preprocessing step, we remove the mean from each pixel and vectorize the images, thus obtaining a graph signal matrix $\bbY \in \reals^{N \times T}$, with $N = 324$. As GSO $\bbS$, we consider the normalized Laplacian matrix of the  $18 \times 18$ grid graph, for which the $5$ exhibits a non-stationary behavior, since $\rho = 0.4$.

\sn
We assume that each $\bby_t$ is the output of a NV-GF $\bbH_I(\bbP, \bbS)$ when excited with a white input $\bbx_t \sim \ccalN(\bb0, \bbI)$; our goal is then to learn the filter taps $\bbP$ and subsequently the dual eigenvalues $\bblambda_f$ from the available data $\bbY$. In this particular scenario, it is important to highlight that our access is limited to the output data $\bbY$, since the corresponding noise input $\bbX$ is not available. Consequently, the sole viable approach becomes the output-only procedure [cf. Section~\ref{subsec:learn-P-output-only}]. Nevertheless, to navigate this limitation and to be also able  to use the input-output approach of Section~\ref{subsec:learn-P}, we can employ the ensuing rationale to derive pairs $(\bbx_t, \bby_t)$:

\begin{enumerate}
    \item[(a)] compute the covariance matrix $\bbC_y = \bbY\bbY^\top 
    / T$ and decompose it as $\bbC_y = \bbR\bbR^\top$;

    \item[(b)] filter 
 a white input signal $\bbx_t$ to generate a new signal $\bby_t$ as $\bby_t = \bbR \bbx_t$.
\end{enumerate}
It follows that the $\bby_t$ vector  generated in this way follow the same distribution of the original $\bbY$ and still represent the digit $5$ (up to a sign ambiguity). This time, however, we have the associated input $\bbX$. We run the proposed algorithm  for different orders $L$ and $K$,  with and without input $\bbX$.

\sn 
 In Fig.~\ref{fig: MNIST}(left) we show the inferred filter taps $\tilde{\bbP}$, for $L=4$, corresponding to the solution of problem~\eqref{eq:main-problem-approximated} and yielding a $\text{NSE}= 3 \times 10^{-2}$; each filter tap $\{\bbp_l\}_{l=0}^3$ has been reshaped to have the same size of the input image and stacked in a row-wise fashion. It is interesting to notice how each $\bbp_l$ has a digit-shape look, with a decreasing pixel intensity for increasing filter tap order $L$; this indicates pixel-locality as an important factor for the creation of the final pixel intensity. Since each (reshaped) filter tap image in Fig~\ref{fig: MNIST} pointwise-multiplies a shifted white input (noise) image of same size (to be finally aggregated), it is visible how the most influential filter content gathers around the digit shape for the digit-formation, and decreases its importance for higher shifts, meaning that the process is local on the pixel and there are no long term-influences. The same NSE and $\bbP$-profile is obtained with the input-output approach.

 \begin{figure*}
    \begin{subfigure}{0.33\textwidth}
        \includegraphics[scale=0.6, trim = 3cm 1cm 3cm 1cm, clip=true]{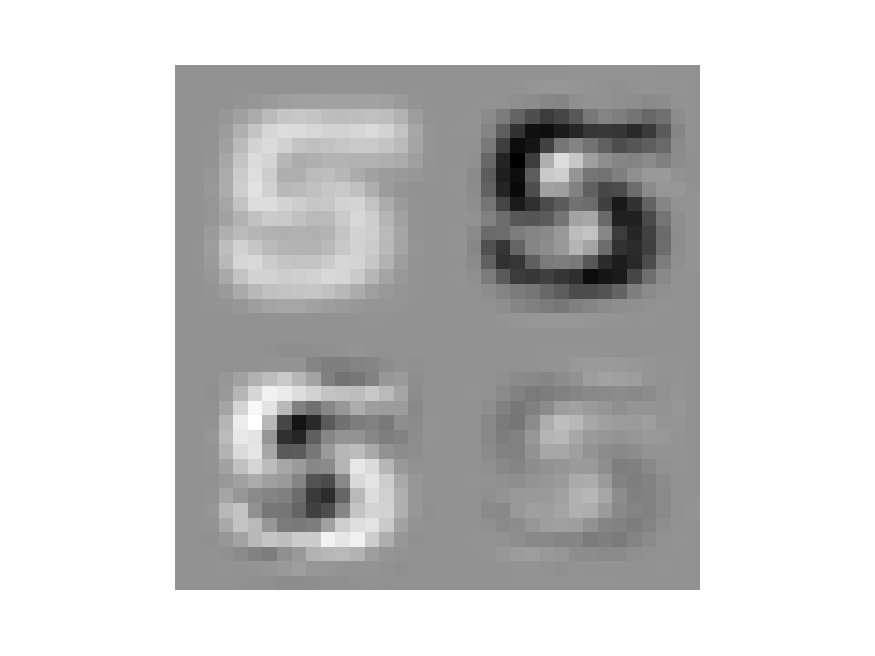}
    \end{subfigure}%
    \begin{subfigure}{0.66\textwidth}
        \includegraphics[scale=0.4,  trim=6cm 3.5cm 5cm 2cm, clip=true]{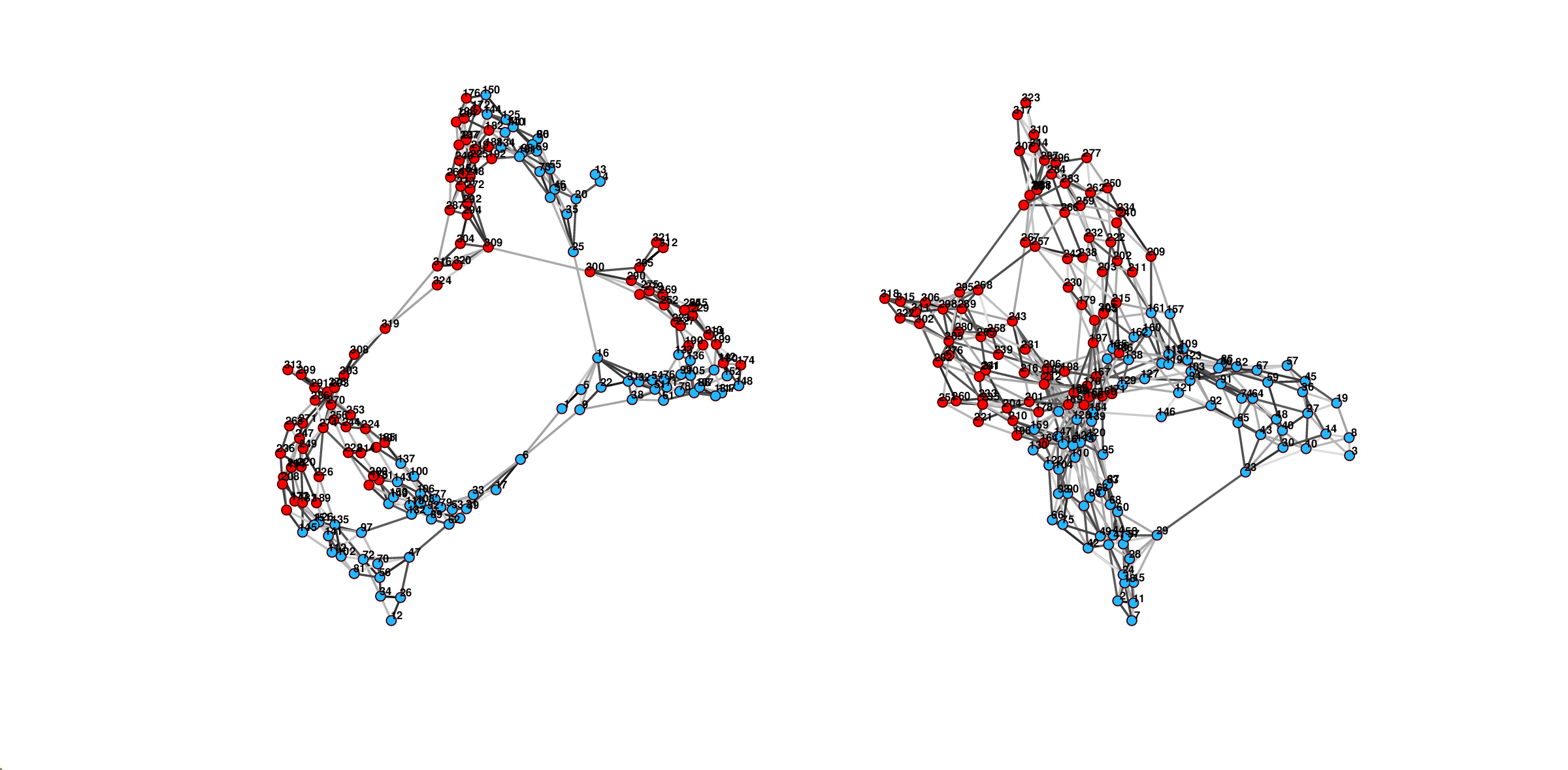}
    \end{subfigure}%
    \vspace{-.2cm}
    \caption{ (Left) Illustration of each filter tap $\{\bbp_l\}_{l=0}^3$, following a row-wise stacking, reshaped to have the same shape of the input image. Increasing $l$ renders the entries of $\bbp_l$ less influential for the construction of the final digit; (Right) Dual Graph of the $5$-digit. The label indicates the index of the eigenvalues $\lambda_i$, and the node color indicates whether it is part of the first half (blue) or the second half (red), commonly associated to the low and high pass bands, respectively. Notice that the graph is connected but due to the edge thresholding (only for the visualization) it is split in two.}
    \label{fig: MNIST}
\end{figure*}

\sn 
Once we have the filter tap matrix $\tilde{\bbP}$, we then learn the dual eigenvalues $\bblambda_f$ 
 by solving~\eqref{eq:equivalent-subspace-fitting}. The associated dual GSO $\bbS_f = \bbV^{-1} \Diag(\tilde{\bblambda}_f) \bbV$, achieving a corollary error [cf.~\eqref{eq:corollary-error}] $\varepsilon_c = 9 \times 10^{-2}$, is shown in Fig.~\ref{fig: MNIST}(right), where we only display $2\%$ of the most significant edges in absolute values\footnote{For a graph with $N=324$ nodes there would be more than $50$k edges possible.}. The node label $i$ indicates the index associated to $\lambda_i$. As in the previous experiment, we color the first half of the nodes (representing now the ``low frequency'' eigenvalues) with blue, and the other half with red. 
 Similar results and conclusions can be made by following the input-output approach, where we obtain $\varepsilon_c = 3 \times 10^{-2}$ and the graph is similarly structured as the one in Fig.~\ref{fig: MNIST} (see Supplementary Material).

 \sn
 Together with the previous experiment, the following observations can be made: \begin{itemize}
     \item The connections between the eigenvalues do not follow the linear ordering as assumed by the traditional real-line interpretation (in that case, we would only have red-red and blue-blue connections without interactions); this has consequences when designing graph filters based solely on the value of the $\lambda_i$, since now the concept of ``bands'' needs also to account for the topology.

     \item  Because the dual graph represents the support of the GFT signals, we can now inspect which neighborhood is influential for a particular frequency during a convolution on $\bbS_f$; this was not possible with the standard real-line interpretation, as the convolution operation was a simple pointwise multiplication.

 \end{itemize}

All in all, these results confirm the commutative nature of the two branches of Fig.~\ref{fig:duality}, thus rendering the dual convolution  a preferred approach when $K<L$, or when $\bbS_f$ and/or the GFT signals exhibit sparsity, in addition to delivering an elegant theoretical framework.


\section{Conclusions}
\label{sec:conclusions}

In this work we proposed a convolution theorem which extends the classical convolution theorem in (graph) signal processing and the one related to time-varying filters. More precisely, we illustrate how a convolution in the primal graph domain can be redefined as a distinct convolution in the dual graph (frequency) domain, given a suitable filter parametrization. After illustrating the implications of such theorem in terms of non-stationarity of signals, and generative models thereof, we devise an algorithmic approach based on subspace fitting and non-convex programming techniques to learn the dual graph from data when this is not a priori known. We evaluated the proposed theory and algorithms on synthetic data, as well as on real data.

While our current theoretical framework holds promise for practical applications in the future, there are notable challenges that merit further exploration. A significant gap lies in the absence of a one-shot procedure to construct the dual GSO directly from a primal GSO, along with potential graph signals associated with it. This limitation curtails the broader applicability of the proposed theoretical insights. Nonetheless, multiple extensions of this work are possible. From an algebraic point of view, an interesting line of research would involve exploring the connections between Vandermonde and Hankel matrices, as well as with Krylov subspaces, potentially unveiling  new algorithmic solutions to learn the dual eigenvalues. From a modeling point of view, an interesting extension of this work would include the node-varying graph filter coefficients also to be time-varying; holding promise for utilization in graph autoregressive models. In such cases, leveraging the basis expansion model technique across the temporal dimension becomes a potential avenue for further investigation. From an optimization point of view, the use of orthogonal polynomials might alleviate the ill-conditioning of the Vandermonde matrix.

Our hope is that in the coming years, further exploration and refinement of this research direction will reveal new insights and methodologies to process signals defined on graphs in a way previously unfeasible.




 
%

\bibliographystyle{IEEEtran}
\bibliography{refs}
\newpage

 




\vfill

\end{document}


\title{Supplemental Material}

\author{Alberto Natali,~\IEEEmembership{Student Member,~IEEE},  Geert Leus,~\IEEEmembership{Fellow,~IEEE}}

\markboth{Journal of \LaTeX\ Class Files,~Vol.~14, No.~8, August~2021}%
{Shell \MakeLowercase{\textit{et al.}}: A Sample Article Using IEEEtran.cls for IEEE Journals}


\maketitle

This document is the accompanying material of the paper ``A Generalization of the Convolution Theorem and its Connections to Non-Stationarity and the Graph Frequency Domain''. 
It includes proofs and detailed mathematical steps that might be viewed as supplementary or beyond the scope of the initial submission guidelines. However, should the reviewers deem them valuable, we are open to incorporating them into the finalized version. The code can be found at~\url{https://github.com/albertonat/genConv}.

\section{Interpretation of Problem 1 as a Sampling Problem}
Illustration of the problem: a number of functions ($L=3$) are sampled (filled points) in $N=8$ locations (gray dotted lines). The problem is to recover the sampling locations and the polynomial coefficients with as only knowledge the function values, the function they belong to, and ordered such that the sampling points for all functions are aligned (shown in the right in the figure at $\lambda=5$).

    \begin{figure}[h!]
    \centering
   \includegraphics[width=0.8\columnwidth, keepaspectratio=true,  trim=1cm 0 1.5cm 0.4cm, clip=true]{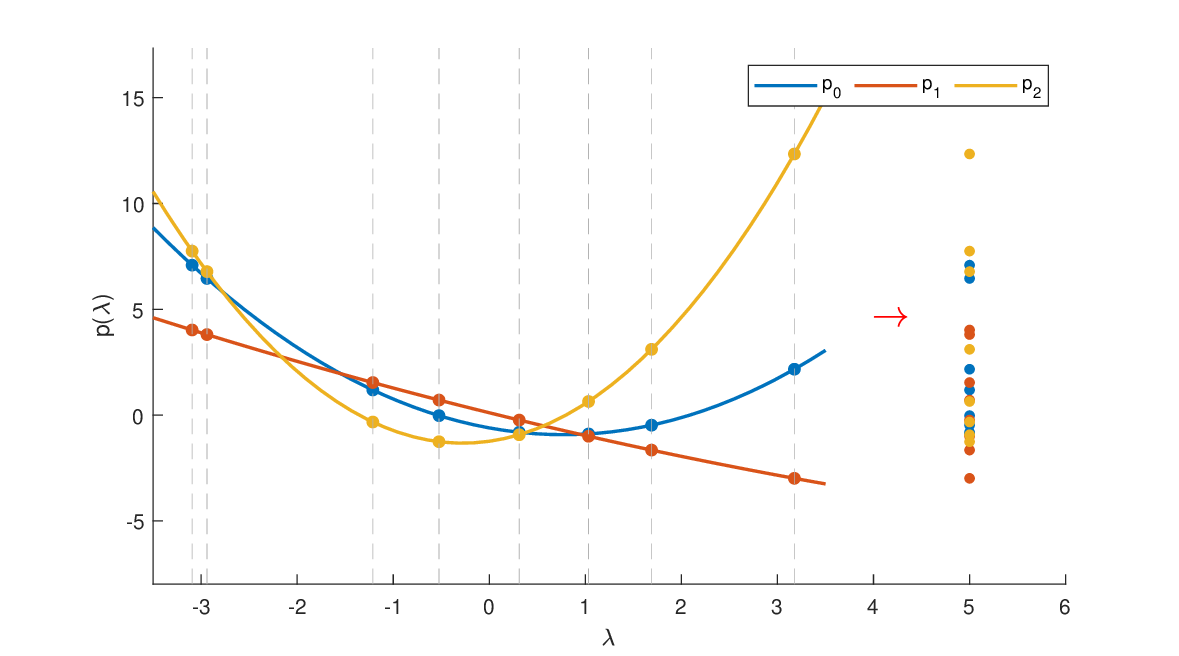}
    \label{fig:idea}
\end{figure}

\section{Sequential Convex Programming}

The function we want to optimize is:
\begin{align}
\label{eq:equivalent-subspace-fitting}
    \min_{\bblambda_f}\; \{ f(\bblambda_f):=\frac{1}{2} \|\boldsymbol{\Pi}\bbPsi_f(\bblambda_f)\|_F^2\},
\end{align}
with $\boldsymbol{\Pi}:=\bbI_N - \bbU\bbU^\dagger$ the orthogonal projection matrix onto the orthogonal complement  of $\bbU$.
The general idea of SCP is to maintain, at each iteration $r$, an estimate $\bblambda_f^{r}$ and a respective convex trust region $\ccalT^{r} \subseteq \reals^{N}$ over which we trust our solution to reside. The next solution $\bblambda_f^{r+1}$ is then obtained by minimizing, over the defined trust region $\ccalT^{r}$, a convex approximation $\hat{f}(\cdot)$ of $f(\cdot)$ around the previous estimate $\bblambda_f^{r}$. In our case, we define as trust region the set:
\begin{align}
\label{eq:trust-region}
    \ccalT^{r}:=\{\bblambda_f \;  \vert \;  \|  \bblambda_f - \bblambda_f^{r}\|_p^2 \leq \rho(r) \}
\end{align}
where $\rho: \mathbb{Z}_{+} \to \reals_{++}$ is an iteration-dependent mapping indicating the maximum admissible length of the convex $p$-norm ball  in~\eqref{eq:trust-region}.

Next, a feasible intermediate iterate is found by:
\begin{align}
    \hat{\bblambda}_f = \argmin_{\bblambda_f \in \ccalT^{[r]}} f(\bblambda_f^{r}) + \nabla_{\bblambda_f} f(\bblambda_f^{r})^\top (\bblambda_f - \bblambda_f^{r})
\end{align}
i.e., by minimizing the (convex) first-order Taylor approximation of $f(\cdot)$ around $\bblambda_f^{r}$, with
 $\nabla_{\bblambda_f} f(\cdot) \in \reals^N$  the gradient of function $f(\cdot)$; see Section~\ref{sec:gradient} for the explicit gradient computation.
Notice that due to the non-convexity of $f(\cdot)$ [cf.~\eqref{eq:equivalent-subspace-fitting}], its value $f( \hat{\bblambda}_f)$ at the new feasible estimate $ \hat{\bblambda}_f$ is not guaranteed to be lower than the value at $\bblambda_f^{r}$. Thus, to find the optimal solution at iteration $r+1$, we pick the best convex combination of $\bblambda_f^{r}$ and $\hat{\bblambda}_f$ achieving the lowest function value, i.e.,:
\begin{align}
\alpha_r^\star &= \argmin_{\alpha_r} f (\alpha_r \bblambda_f^{r} + (1-\alpha_r)\hat{\bblambda}_f), \quad \alpha_r \in (0,1) \\
    \bblambda_f^{r+1}&=\alpha_r^\star \bblambda_f^{r} + (1-\alpha_r^\star)\hat{\bblambda}_f.
\end{align}
We stress that a zero fitting error of the cost function can be achieved only by the true unknown parameter and its model ambiguities, which as such represent global minima of the function and solutions of the problem.

\section{Computation of the gradient}
\label{sec:gradient}

 Consider the function $f(\bblambda_f)= \|\boldsymbol{\Pi}\bbPsi_f(\bblambda_f)\|_F^2$ and define the composite matrix function  $\bbM(\bbPsi_f(\bblambda_f))= \boldsymbol{\Pi}\bbPsi_f(\bblambda_f)$. We use the chain rule to compute the gradient of $f(\cdot)$ w.r.t. $\bblambda_f$, i.e., $\partial f (\cdot)/ \partial \bblambda_f $.
We have $ f(\bbM)= \frac{1}{2} \tr(\bbM^\top \bbM)= \frac{1}{2} \bbM : \bbM$. The differential of $f(\cdot)$ is $df(\bbM, d\bbM)= \bbM : d\bbM= \bbM : \boldsymbol{\Pi} d\bbPsi_f(\bblambda_f)$. For the differential $d\bbPsi_f(\bblambda_f, d\bblambda_f)$ notice that, by selecting one column, e.g., $\bblambda_f^3$, we have $(\bblambda_f + d\bblambda_f)^3= \bblambda_f^3 + 3 \bblambda_f^2 d\bblambda_f + o(\|d\bblambda_f\|_2^2)$, and more generally:
\begin{align}
\nonumber
 d\bbPsi_f (\bblambda_f) &= [\boldsymbol{0} \; \boldsymbol{1} \ldots (K-1)\bblambda_f^{K-2}]  \odot [d\bblambda_f \;  d\bblambda_f \ldots d\bblambda_f]\\ \nonumber
&= (\bbPsi_f(\bblambda_f) \bbD_k) \odot (d\bblambda_f \boldsymbol{1}^\top) \\
&= \Diag(d\bblambda_f)(\bbPsi_f(\bblambda_f) \bbD_k)
\end{align}
where $\odot$ is the element-wise product and $\bbD_k$ is the differentiation matrix for polynomials, i.e.,:
\begin{align}
    \bbD_k= \begin{bmatrix}
         0 & 1 & 0 & 0 & \cdots & 0 \\
         0 & 0 & 2 & 0 & \cdots & 0 \\
         0 & 0 & 0 & 3 & \cdots & 0 \\
         \vdots & 0 & 0 & 0 & \cdots & \vdots\\
          0 & 0 & 0 & 0 & \cdots & K-1 \\
          0 & 0 & 0 & 0 & \cdots & 0
    \end{bmatrix} \in \reals^{K \times K}.
\end{align}
By plugging this differential into the previous expression for $df$ we obtain:
\begin{align*}
    df(\bbM,d\bbM)&= \bbM :  \boldsymbol{\Pi} d\bbPsi_f(\bblambda_f) \\&=\bbM :  \boldsymbol{\Pi}  \Diag(d\bblambda_f)(\bbPsi_f(\bblambda_f) \bbD_k) 
    \\&=\boldsymbol{\Pi}^\top \bbM (\bbPsi_f(\bblambda_f) \bbD_k)^\top :    \Diag(d\bblambda_f)  
     \\&=\diag(\boldsymbol{\Pi}^\top \bbM (\bbPsi_f(\bblambda_f) \bbD_k)^\top) :    d\bblambda_f 
\end{align*}
We conclude 
\begin{align}
    \partial f(\bblambda_f)/ \partial \bblambda_f= \diag(\boldsymbol{\Pi}^\top \boldsymbol{\Pi}\bbPsi_f(\bblambda_f)\bbD_k^\top \bbPsi_f(\bblambda_f)^\top).
\end{align}

\newpage
\section{Illustration of the Synthetic Numerical Workflow}

Schematic representation of the simulation's workflow, constituting of a  generation and a learning phases.
\begin{figure}[h!]
\centering
\includegraphics[trim= 2.5cm 6.7cm 3.5cm 5.3cm, clip=true, scale= 0.7]{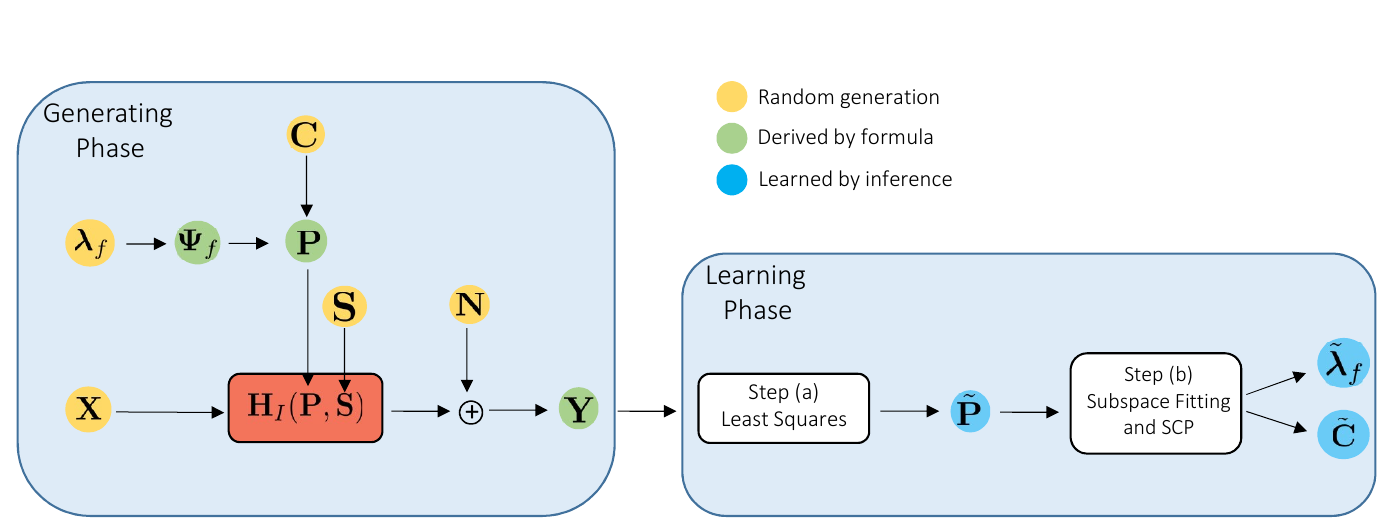}
\label{fig:generation_learning_phases}
\end{figure}

\section{MNIST Dataset}
Below we show the dual graphs obtained with the input-output approach (LS) and with the output-only approach. Notice how despite the objective function of the two problems is different, they lead to almost identical dual graphs in terms of connections. Notice that we show only the $2\%$ biggest edges (in absolute values), reason why the graph is split in two (although the non-split graph is connected).

\begin{figure}[h!]
    \centering
    \includegraphics[scale=0.6,  trim=6cm 3.5cm 5cm 2cm, clip=true]{figures/real/dual_graph.eps}
        \includegraphics[scale=0.4, trim=4cm 2cm 4cm 1cm, clip=True]{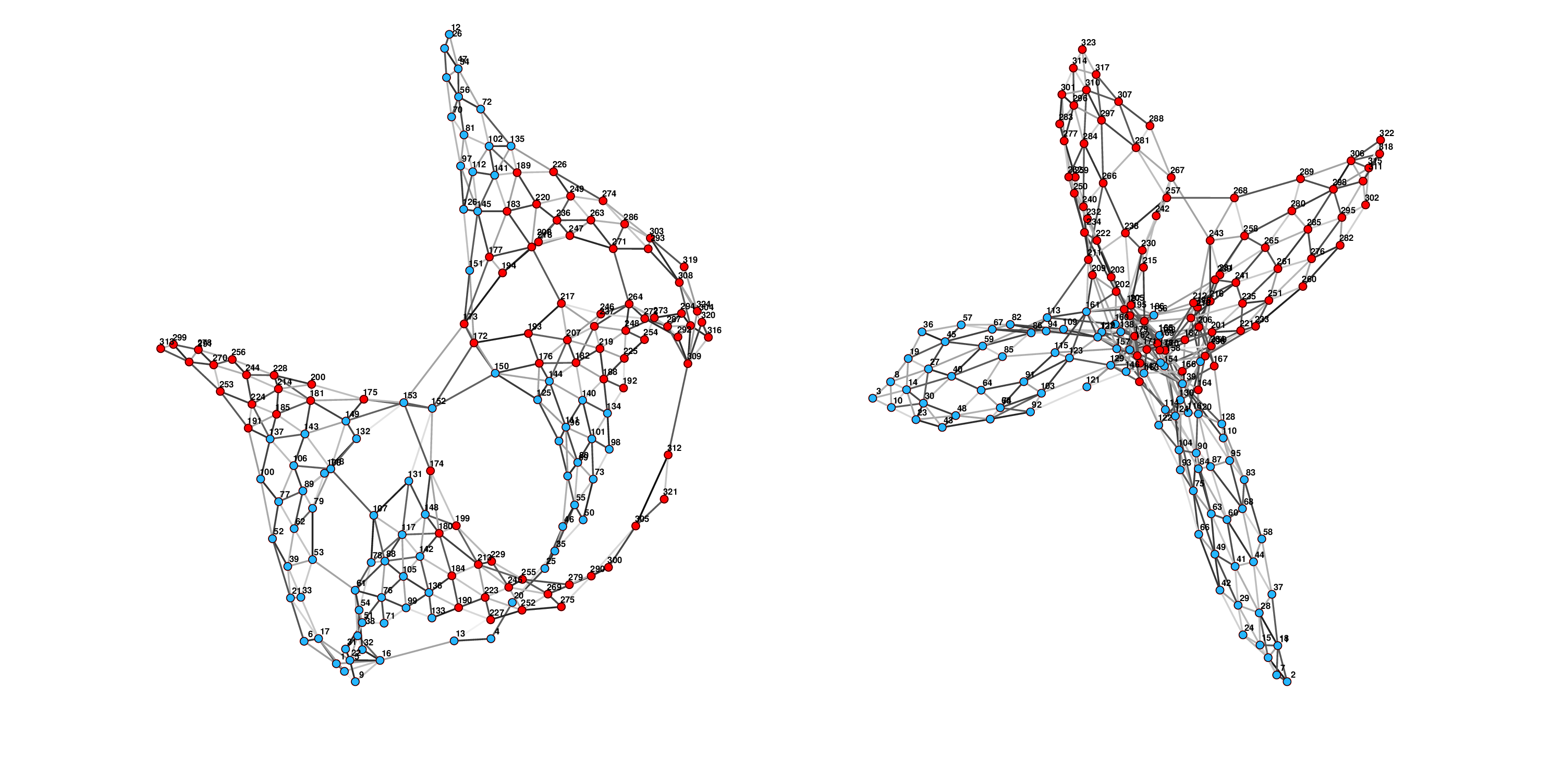}
    \caption{(Top) Dual graph obtained with the output-only approach; (Bottom) Dual graph obtained with the input-output approach}
    \label{fig:enter-label}
\end{figure}

\vfill